\pdfoutput=1
\documentclass{article}
\usepackage{amsfonts,amsmath,amssymb}
\usepackage{fullpage}

\usepackage{tabularx,multirow}
\usepackage{booktabs}
\usepackage{xspace}
\usepackage[normalem]{ulem}

\usepackage{url}

\usepackage{gnuplottex}

\usepackage{enumitem}

\newenvironment{smatrix}{\left[\begin{smallmatrix}}{\end{smallmatrix}\right]}

\newcommand{\cmark}{$\checkmark$}%
\newcommand{\xmark}{\text{\sffamily X}}%

\usepackage{arydshln} %
\usepackage[symbol]{footmisc}

\usepackage{array,graphicx}

\usepackage{algorithm}
\usepackage{algorithmic}

\makeatletter
\newcommand{\CELSE}[1][default]{\end{ALC@if}\ALC@it\algorithmicelse%
\hfill\ALC@com{#1}\begin{ALC@if}}
\newcommand{\CRETURN}[1]{\ALC@it #1: \algorithmicreturn{}\xspace}
\makeatother

\usepackage{amsthm}
\usepackage{thm-restate}
\newtheorem{theorem}{Theorem}

\newtheorem{remark}[theorem]{Remark}
\newtheorem{remarks}[theorem]{Remarks}
\newtheorem{definition}[theorem]{Definition}

\newcommand{\Z}{\ensuremath{\mathbb{Z}}}
\newcommand{\N}{\ensuremath{\mathbb{N}}}
\newcommand{\random}{\stackrel{\$}{\leftarrow}}
\newcommand{\checks}[1]{\ensuremath{\mathrel{\stackrel{?}{#1}}}}

\newcommand{\F}{\ensuremath{{\mathbb F}}}
\newcommand{\CC}{\ensuremath{{\mathcal C}}}
\newcommand{\K}{\ensuremath{{\mathcal K}}}

\newcommand{\mtinit}{\textbf{MTInit}\xspace}
\newcommand{\mtread}{\textbf{MTVerifiedRead}\xspace}
\newcommand{\mtwrite}{\textbf{MTVerifiedWrite}\xspace}
\newcommand{\mtblock}{\ensuremath{b}}
\newcommand{\mtnbblock}{\ensuremath{\lceil{}N/\mtblock\rceil{}}}
\newcommand{\E}{\ensuremath{\mathbb{E}}}
\newcommand{\clientinstance}{\ensuremath{C}}
\newcommand{\serva}{\ensuremath{S_1}}
\newcommand{\servb}{\ensuremath{S_4}}
\newcommand{\servc}{\ensuremath{S_{16}}}

\newcommand{\vect}[1]{\ensuremath{\mathbf{#1}}}
\newcommand{\matr}[1]{\ensuremath{\mathbf{#1}}}
\newcommand{\uu}{\vect{u}}
\def\vv{\vect{v}} %
\newcommand{\xx}{\vect{x}}
\newcommand{\yy}{\vect{y}}
\newcommand{\zz}{\vect{z}}
\newcommand{\WW}{\matr{W}}
\newcommand{\svec}{\vect{s}}
\newcommand{\MM}{\matr{M}}
\newcommand{\UU}{\matr{U}}
\newcommand{\VV}{\matr{V}}
\newcommand{\XX}{\matr{X}}
\newcommand{\YY}{\matr{Y}}
\newcommand{\AAA}{\matr{A}}
\newcommand{\BB}{\matr{B}}
\newcommand{\tp}{^\intercal}

\newcommand{\ExternT}{extern=\texttt{T}}
\newcommand{\ExternF}{extern=\texttt{F}}

\newcommand{\compsec}{\ensuremath{\kappa}}
\newcommand{\statsec}{\ensuremath{\lambda}}

\renewcommand{\arraystretch}{1.2}

\usepackage{color,svgcolor}
\usepackage[plainpages=true]{hyperref}

\makeatletter
\hypersetup{
pdftitle={Dynamic proofs of retrievability with low
  server storage},
pdfauthor={Gaspard Anthoine, Jean-Guillaume Dumas, Michael Hanling, M\'elanie de Jonghe, Aude Maignan, Cl\'ement Pernet, Daniel S. Roche},
breaklinks=true,
colorlinks=true,
 linkcolor=darkblue,
 citecolor=darkgreen,
 urlcolor=darkred,
}
\makeatother

\usepackage[capitalise,noabbrev,nameinlink]{cleveref}

\newcommand{\email}[1]{\href{mailto:#1}{\nolinkurl{#1}}}
\newcommand{\emails}[2]{\href{mailto:#2}{\nolinkurl{#1}}}
\begin{document}
\title{Dynamic proofs of retrievability with low server storage}
\title{\Large\bf Dynamic proofs of retrievability with low
  server storage}
\author{Gaspard Anthoine\footnote{Universit\'e Grenoble Alpes,
    Laboratoire Jean Kuntzmann, UMR CNRS 5224, Grenoble INP. 700
    avenue centrale, IMAG-CS 40700, 38058 Grenoble, France.
    \email{Gaspard.Anthoine@etu.univ-grenoble-alpes.fr},
    \emails{{Jean-Guillaume.Dumas,Aude.Maignan,Clement.Pernet}@univ-grenoble-alpes.fr}{Jean-Guillaume.Dumas@univ-grenoble-alpes.fr,Aude.Maignan@univ-grenoble-alpes.fr,Clement.Pernet@univ-grenoble-alpes.fr},
    \email{dejonghe.melanie63@gmail.com}.
}
\and Jean-Guillaume Dumas\footnotemark[1]
\and Michael Hanling\footnote{United States Naval Academy,
  Annapolis, Maryland, United States.
  \email{mikehanling@gmail.com},
  \email{roche@usna.edu}
}
\and M\'elanie de Jonghe\footnotemark[1]
\and Aude Maignan\footnotemark[1]
\and Cl\'ement Pernet\footnotemark[1]
\and Daniel S.\ Roche\footnotemark[2]
}

\maketitle
\begin{abstract}
Proofs of Retrievability (PoRs) are protocols which allow a client to
store data remotely and to efficiently ensure, via audits, that the
entirety of that data is still intact. A \emph{dynamic} PoR system also
supports efficient retrieval and update of any small portion of the
data. We propose new, simple protocols for dynamic PoR that are designed
for practical efficiency, trading decreased persistent storage for
increased server computation, and show in fact that this tradeoff is
inherent via a lower bound proof of time-space for any PoR scheme.
Notably, ours is the first dynamic PoR which does not require any
special encoding of the data stored on the server, meaning it can be
trivially composed with any database service or with existing techniques
for encryption or redundancy. Our implementation and deployment on
Google Cloud Platform demonstrates our solution is scalable: for
example, auditing a 1TB file takes just {less than 5 minutes}
and costs less than {\$0.08 USD}. We also present several further enhancements, reducing
the amount of client storage, or the communication bandwidth, or
allowing \emph{public verifiability}, wherein any untrusted third party
may conduct an audit.
\end{abstract}

\section{Introduction}
While various computing metrics have accelerated and slowed over the last
half-century, one which undeniably continues to grow quickly is
data storage. One recent study estimated the world's storage capacity at
4.4ZB ($4.4\cdot{}10^{21}$), and growing at a rate of 40\% per year
\cite{StoragesAreNotForever2017}. Another study %
group estimates that by 2025, half of the world's data will be stored
remotely, and half of that will be in public cloud storage
\cite{DataAge2025}.

As storage becomes more vast and more outsourced, users and
organizations need ways to ensure the \emph{integrity} of their data --
that the service provider continues to store it, in its entirety, unmodified.
Customers may currently rely on the reputations of large cloud companies
like IBM Cloud or Amazon AWS, but even those can suffer
data loss events \cite{aws-loss-2019,google-loss-2015},
and as the market continues to grow, new
storage providers without such long-standing reputations need
cost-effective ways to convince customers their data is intact.

This need is especially acute for the growing set of
\emph{decentralized storage networks} (DSNs),
such as
\href{https://filecoin.io/filecoin.pdf}{Filecoin},
\href{https://storj.io/}{Storj},
\href{https://sia.tech/}{Sia},
\href{https://safenetwork.tech/}{SAFE Network},
and \href{https://www.pp.io/}{PPIO}%
\footnote{
\url{https://filecoin.io},
\url{https://storj.io},
\url{https://sia.tech},
\url{https://safenetwork.tech},
\url{https://www.pp.io}.}%
,
that act to connect
users who need their data stored with providers
(``miners'') who will be paid to store users' data. In DSNs,
integrity checks are useful at two levels: from the customer
who may be wary of trusting blockchain-based networks, and within
the network to ensure that storage nodes are actually providing their
promised service.  Furthermore, storage nodes whose sole aim is to earn
cryptocurrency payment have a strong incentive to cheat, perhaps by
deleting user data or thwarting audit mechanisms.

The research community has developed a wide array of solutions to the
remote data integrity problem over the last 15 years. Here we merely
summarize the main lines of work and highlight some shortcomings that
this paper seeks to address.

\paragraph{Provable Data Possession (PDP).}
PDP audits \cite{juels2007pors,Erway,stefanov2012iris} are
efficient methods to ensure that a large fraction of data
has not been modified. They generally work by computing a small
\emph{tag} for each block of stored data, then randomly sampling a
subset of data blocks and corresponding tags, and computing a check over
that subset.

Because a server that has lost or deleted a constant fraction of
the file will likely be unable to pass an audit, PDPs are useful in
detecting catastrophic or unintentional data loss. They are also quite
efficient in practice. However, \emph{a server who deletes only a few
blocks is still likely to pass an audit}, so the security guarantees are
not complete, and may be inadequate for critical data storage or
possibly-malicious providers.

\paragraph{Proof of Retrievability (PoR).}

PoR audits, starting with \cite{ateniese2007provable},
have typically used techniques such as error-correcting codes,
and more recently Oblivious RAM (ORAM), in order to obscure from
the server where pieces of the file are stored
\cite{Lavauzelle:2016:ldcpor,DVW-por-09}.
Early PoR schemes did not provide an efficient update mechanism to alter
individual data blocks, but more recent \emph{dynamic} schemes have
overcome this shortcoming~\cite{Shi:2013:orampor,Cash:2017:DPR}.

A successful PoR audit
provides a strong guarantee of retrievability: if the server altered
many blocks, this will be detected with high probability, whereas if
only few blocks were altered or deleted, then the error correction means
the file can still likely be recovered. Therefore, a single successful
audit ensures with high probability that the \emph{entire} file is still
stored by the server.

The downside of this stronger guarantee is that PoRs have
typically used more sophisticated cryptographic tools than PDPs,
and in all cases we know of
require \emph{multiple times the original data size for persistent
remote storage}. This is problematic from a cost standpoint: if a PoR
based on ORAM requires perhaps 10x storage on the cloud, this cost
may easily overwhelm the savings cloud storage promises to provide.

For our purpose, we have identified two main storage outsourcing type
of approaches: those which minimizes the storage overhead and
those which minimize the client and server computation. For each
approach, we specify in Table~\ref{tab:spec} which one  meets
various requirements such as whether or not they are dynamic, if they
can answer an unbounded number of queries and what is the extra
storage they require.

\begin{table}[htbp]\centering\small
\renewcommand{\arraystretch}{.75} %
\caption{Attributes of some selected schemes} \label{tab:spec}
\begin{tabular}{lcccc}
\toprule
         & PoR &\multicolumn{2}{c}{Number of} & Extra \\
Protocol & capable &audits & updates&  Storage\\
\midrule
Seb\'e \cite{Sebe:2008:EfficientRD} & \xmark & {\Large $\infty$}& \xmark &$o(N)$ \\
Ateniese et al. \cite{ateniese2007provable} &\xmark & {\Large $\infty$}& \xmark & $o(N)$ \\
Ateniese et al.  \cite{ateniese2008scalable} & \xmark& $O(1)$ & $O(1)$ &$o(N)$\\
Storj \cite{Storj:2016:por}& \cmark& $O(1)$ &{\Large $\infty$} & $o(N)$ \\
Juels et al. \cite{juels2007pors} & \cmark & $O(1)$ &  \xmark&  $O(N)$ \\
Lavauzelle et al. \cite{Lavauzelle:2016:ldcpor} & \cmark & {\Large $\infty$} & \xmark & $O(N)$  \\
Stefanov et al. \cite{stefanov2012iris}& \cmark & {\Large $\infty$} & {\Large $\infty$} &$O(N)$ \\
Cash et al. \cite{Cash:2017:DPR}& \cmark & {\Large $\infty$} & {\Large $\infty$} &$O(N)$ \\
Shi et al. \cite{Shi:2013:orampor}& \cmark & {\Large $\infty$} & {\Large $\infty$} &$O(N)$ \\
Here & \cmark & {\Large $\infty$} & {\Large $\infty$} & $o(N)$  \\
\bottomrule
\end{tabular}
\end{table}
\Cref{sec:sota} gives a detailed comparison with prior work.
\paragraph{Proof of Replication (PoRep) and others.}

While our work mainly falls into the PoR/PDP setting, it also has
applications to more recent and related notions of remote storage
proofs.

Proofs of space were originally proposed as an alternative to the
computation-based puzzles in blockchains and anti-abuse mechanisms
\cite{ABFG-pos-2014,DFKP-pos-2015}, and require verifiable storage of a
large amount of essentially-random data.
A PoRep scheme (sometimes called \emph{Proof of Data Reliability})
aims to combine the ideas of proof of space and PoR/PDP
in order to prove that \emph{multiple copies} of a
data file are stored remotely. This is important as, for example, a
client may pay for 3x redundant storage to prevent data loss, and wants
to make sure that three actual copies are stored in distinct locations.
Some PoRep schemes employ slow encodings and time-based audit checks;
the idea is that a server does not have enough time to re-compute the
encoding on demand when an audit is requested, or even to retrieve it
from another server, and so must actually
store the (redundantly) encoded file
\cite{Armknecht:2016:mirror,fisch-porep,Vasilopoulos:2019:portos,Cecchetti:2019:PIE}.
The Filecoin network employs this type of verification.
A different and promising approach, not based on timing assumptions, has recently been
proposed by \cite{DGO-porep-19}.
An important property of many recent PoRep schemes is \emph{public
verifiability}, that is, the ability for a third party (without secrets)
to conduct an audit. This is crucial especially DSNs.

Most relevant for the current paper is that \emph{most of these schemes
directly rely on an underlying PDP or PoR} in order to verify encoded
replica storage. For example, \cite{DGO-porep-19} states that their
protocol directly inherits any security and efficiency properties of the
underlying PDP or PoR.

We also point out that, in contrast to our security model, many of these
works are based on a \emph{rational actor model}, where it is not in a
participant's financial interest to cheat, but a malicious user may
break this guarantee, and furthermore that most existing PoRep schemes
do not support \emph{dynamic} updates to individual data blocks.

\subsection{Our Contributions}

We present a new proof of retrievability which has the following
advantages compared to existing PDPs and PoRs:
  \par{\bf Near-optimal persistent storage.}
    The best existing PoR protocols that we could find require between
    $2N$ and $10N$
    bytes of cloud storage to support audits of an $N$-byte
    data file, making these schemes impractical in many settings. Our
    new PoR requires only $N + O(N/\log N)$ persistent storage.
  \par{\bf Simple cryptographic building blocks.}
    Our basic protocol relies only on small-integer arithmetic and a
    collision-resistant hash function, making it very efficient in
    practice. Indeed, we demonstrate that 1TB of data can be
    audited in {less than $5$ minutes at a monetary cost of just $0.08$ USD}.
  \par{\bf Efficient partial retrievals and updates.}
    That is, our scheme is a \emph{dynamic} PoR, suitable to large
    applications where the user does not always wish to re-download the
    entire file.
  \par{\bf Provable retrievability from malicious servers.}
    Similar to the best PoR protocols, our scheme supports data recovery
    (\emph{extraction}) via rewinding audits. This means, in particular,
    that there is only a negligible chance that a server can pass a
    \emph{single} audit and yet not recover the entirety of stored data.
  \par{\bf (Nearly) stateless clients.}
    With the addition of a symmetric cipher, the client(s) in our
    protocol need only store a single decryption key and hash digest,
    which means multiple clients may easily share access (audit
    responsibility) on the same remote data store.
  \par{\bf Public verifiability.}
    We show a variant of our protocol, based on the difficulty of
    discrete logarithms in large groups, that allows any third party to
    conduct audits with no shared secret.

Importantly, because our protocols store the data unencoded on the
server, they can trivially be used within or around any existing
encryption or duplication scheme, including most PoRep constructions.
We can also efficiently support arbitrary server-side applications, such
as databases or file systems with their own encoding needs.
The main drawback of our schemes is that, compared to existing PoRs,
they have a higher asymptotic complexity for server-side computation
during audits, and (in some cases) higher communication bandwidth
during audits as well. However, we also provide a time-space lower bound
that proves \emph{any PoR scheme} must make a tradeoff between
persistent space and audit computation time.

Furthermore, we demonstrate with a complete implementation and
deployment on Google Compute Platform that \emph{the tradeoff we make is
highly beneficial in cloud settings}. Intuitively, a user must pay for
the computational cost of audits only when they are actually happening,
maybe a few times a day, whereas the extra cost of (say) 5x persistent
storage \emph{must be paid all the time}, whether the client is
performing audits or not.

\subsection{Organization}

The rest of the paper is structured as follows:
    \Cref{sec:model} defines our security model, along the lines of
    most recent PoR works;
    \Cref{sec:lowerbound} contains our proof of an inherent
    time-space tradeoff in any PoR scheme;
    \Cref{sec:simple} gives an overview and description of our basic
    protocol, with detailed algorithms and security proofs delayed
    until \cref{sec:formal};
    {the latter thus presents the formal setting, and}
    also contains
    a publicly verifiable variant;
    \Cref{sec:impl} discusses the results of our open-source implementation and
    deployment on Google Compute Platform.
\section{Security model}
\label{sec:model}
\newcommand{\ttb}[1]{\ensuremath{\text{\bf\texttt{#1}}}}
\newcommand{\client}{\ensuremath{\mathcal{C}}}
\newcommand{\server}{\ensuremath{\mathcal{S}}}
\newcommand{\clstate}{\ensuremath{st_\client}}
\newcommand{\pubclstate}{\ensuremath{pub_\client}}
\newcommand{\prvclstate}{\ensuremath{prv_\client}}
\newcommand{\servstate}{\ensuremath{st_\server}}
\newcommand{\poly}{\mathsf{poly}}
\newcommand{\negl}{\mathsf{negl}}
\newcommand{\db}{\ensuremath{M}}
\newcommand{\Init}{\ttb{Init}}
\newcommand{\Read}{\ttb{Read}}
\newcommand{\Write}{\ttb{Write}}
\newcommand{\Audit}{\ttb{Audit}}
\newcommand{\Extract}{\ttb{Extract}}
\newcommand{\IsValid}{\ttb{IsValid}}
\newcommand{\accept}{\ttb{accept}}
\newcommand{\reject}{\ttb{reject}}
\newcommand{\fail}{\ttb{fail}}
\newcommand{\nextract}{\ensuremath{e}}
\newcommand{\nextractb}{\ensuremath{e}}
\newcommand{\adversary}{{\ensuremath{\overline{\server}}}}
\newcommand{\observer}{\ensuremath{\mathcal{O}}}

We define a dynamic PoR scheme as consisting of the following five
algorithms between a client \client{} with state \clstate{} and
a server \server{} with state \servstate{}. Our definition is the same
as given by \cite{Shi:2013:orampor}, except that we follow
\cite{juels2007pors} and include the
\Extract{} algorithm in the protocol explicitly.

A subtle but important point to note is that, unlike the first four
algorithms, \Extract{} is not really intended to be used in practice.
In typical usage, a cooperating and honest server will pass all audits,
and the normal \Read{} algorithm would be used to retrieve any or all of
the data file reliably. The purpose of \Extract{} is mostly to prove
that the data is recoverable by a series of random, successful audits,
and hence that the server which has deleted even one block of data has
negligible chance to pass a single audit.

Our definitions rely on two distinct security parameters,
\compsec{} for computational security and \statsec{} for statistical
security.
Typically values of $\compsec \ge 128$ and
$\statsec \ge 40$ are considered secure \cite{Evans:2018:SMCintro}.
One may think of \compsec{} having to do with offline attacks
and \statsec{} corresponding only to online attacks which require
interaction and where the adversary is more limited.
Carefully tracking both security parameters in our analysis will
allow us to more tightly tune performance without
sacrificing security.

  The {s}erver computation in all these algorithms is
  deterministic while
the client may use random coins for any
algorithm; at a minimum, the \Audit{} algorithm \emph{must} be
randomized in order to satisfy retrievability non-trivially.

\begin{itemize}[nosep]
\item $(\clstate, \servstate)\leftarrow\Init(1^\compsec,1^\statsec,\mtblock,\db)$:
  On input of the security parameters and the database $\db$,
  consisting of $N$ bits arranged in blocks of $\mtblock$ bits,
  outputs the client state $\clstate$ and the server state $\servstate$.
\item $\{m_i,\reject\}\leftarrow\Read(i, \clstate,
  \servstate)$: On input of an index $i\in{}1..\mtnbblock$, the client state
  $\clstate$ and
  the server state $\servstate$, outputs $m_i = \db[i]$ or \reject{}.
\item $\{(\clstate',
  \servstate'),\reject\}\leftarrow\Write(i, a, \clstate,
  \servstate)$: On input of an index $i\in{}1..\mtnbblock$, data $a$, the client
  state $\clstate$ and the server state $\servstate$, outputs a
  new client state $\clstate'$ and a new server state $\servstate'$,
  such that now $\db[i]=a$, or \reject{}.
\item $\{\pi,
  \reject\}\leftarrow\Audit(\clstate, \servstate)$ : On
  input of the client state $\clstate$ and the server state $\servstate$,
  outputs a successful transcript \(\pi\) or \(\reject\).
\item \(\db \gets \Extract(\clstate, \pi_1, \pi_2, \ldots, \pi_\nextract)\):
  On input of independent \Audit{} transcripts
  \(\pi_1,\ldots,\pi_\nextract\), outputs the database \db{}.
  The number of required transcripts \nextract{} must be a polynomially-bounded function of
  \(N\), \(\mtblock\), and \(\compsec\).
\end{itemize}

\subsection{Correctness}
A correct execution of the algorithms by {an} honest client
and {an} honest server
results in audits being accepted and reads to recover the last updated
value of the database. More formally, correctness~is:
\begin{definition}[Correctness]\label{def:corr}
  For any parameters \(\compsec,\statsec,N,\mtblock\), there exists a predicate
  \IsValid{} such that, for any database \(\db\) of \(N\) bits,
  \(\IsValid(\db, \Init(1^\compsec,1^\statsec,\mtblock,\db))\). Furthermore, for any state
  such that \(\IsValid(\db,(\clstate,\servstate))\) and any index
  \(i\) with \(0\le i<\mtnbblock\), we have
  \begin{itemize}[nosep]
    \item \(\Read(i, \clstate, \servstate) = \db[i]\);
    \item \(\IsValid(\db',\Write(i,a,\clstate,\servstate))\),
      where \(\db'[i]=a\) and the remaining \(\db'[j]=\db[j]\) for every \(j\ne i\);
    \item \(\Audit(\clstate,\servstate) \ne \reject\);
    \item For \nextract{} audits \(\Audit_1,\ldots,\Audit_\nextract\)
      with independent randomness,
      with probability \(1 - \negl(\statsec)\):
  \begin{multline*}
\Extract(\clstate,\Audit_1(\clstate,\servstate),\ldots,\\\Audit_\nextract(\clstate,\servstate))=M.
\end{multline*}
  \end{itemize}
\end{definition}

Note that, even though \client{} may use random coins in the algorithms,
a correct PoR by this definition should have no chance of returning
\reject{} in any
\Read{}, \Write{} or \Audit{} with an honest client and server.

\subsection{Authenticity and attacker model}\label{ssec:authenticity}
The authenticity requirement stipulates that
the client can always detect (except with negligible probability) if
any message sent by the server deviates from honest behavior.
To distinguish between public and private verification,
{we consider now two types of client: a \emph{Writer} who
can run any of the \Init{}, \Write{}, \Read{}, or \Audit{} algorithms;
and a \emph{Verifier} that can only run the last two.}
{Accordingly,}
we 
split the client state in two parts: the secret values,
{\prvclstate}, and the published ones,
{\pubclstate, so that
  $\clstate=\pubclstate\cup\prvclstate$.}

In a privately verifiable protocol, detecting such
deviations requires both secret and public values.

In a publicly verifiable protocol we distinguish
the \Init{} and \Write{} algorithms
{(which use the full client state \clstate)}
from the \Read{} and \Audit{} ones {(which use only the
  public part \pubclstate)}.
Detecting a deviation in  \Init{} and \Write{} still require the full
client state (both the secret and public parts), while detecting a
deviation in  \Read{} and \Audit{} must be possible using only the
public parts of the client state.
We use the following game between
two observers $\observer_1$ and $\observer_2$,
a potentially {\em malicious} server $\adversary$ and an honest server $\server$
for the adaptive version of authenticity.
  This is the game of~\cite{Shi:2013:orampor}, generalized to
  exhibit a public/private distinction:

\begin{enumerate}[nosep]
\item $\adversary$ chooses an initial memory $\db$.
$\observer_1$ runs \Init{}, sends {\pubclstate{}} to
$\observer_2$
and sends the initial memory layout $\servstate$
to both~$\adversary$  and~$\server$.
\item For a polynomial number of steps $t = 1, 2, . . . , poly(\statsec)$,
$\adversary$ picks an operation $op_t$ where operation $op_t$ is either
\Read{}, \Write{} or \Audit{}.
$\observer_1$ and $\observer_2$ execute their respective operations
with both~$\adversary$ and~$\server$.
\item $\adversary$ is said to win the game, if any message sent
  by~$\adversary$ differs from that of~$\server$ and
  neither~$\observer_1$, nor~$\observer_2$,
  did output \reject{}.
\end{enumerate}

\begin{definition}[Public Verifiability]\label{def:pubauth}
  A PoR scheme
  satisfies
    \emph{public adaptive authenticity} (or
    \emph{public verifiability}),
  if no polynomial-time
  adversary $\adversary$ has more than negligible probability in
  winning the above security game.
\end{definition}

\begin{definition}[Authenticity]\label{def:privauth}
  A PoR scheme
  satisfies \emph{private adaptive authenticity} (or just
  \emph{adaptive authenticity}), if no polynomial-time adversary
  $\adversary$ has more than negligible probability in winning the
  above security game when $\observer_1$ also plays the role
  of~$\observer_2$.
\end{definition}

\subsection{Retrievability}
Intuitively, the retrievability requirement
stipulates that whenever a malicious server can pass the
audit test with high probability, the server must
know the entire memory contents $\db$.
To model this,~\cite{Cash:2017:DPR} uses a \emph{blackbox rewinding
access}: from the state of the server before any passed audit, there
must exist an extractor algorithm that can reconstruct the complete
correct database.
As in~\cite{Shi:2013:orampor}, we insist furthermore that the extractor
does not use the complete server state, but only the transcripts from
successful audits. In the following game, note that the observer
  \(\observer_1\)
running the honest client algorithms may only update its
state \clstate{} during \Write{} {operations}, and
{that} the \Audit{} {operations} are independently
randomized from the client side, but we make
no assumptions about the state of the adversary \adversary{}.
\begin{enumerate}[nosep]
\item  $\adversary$ chooses an initial database $\db$.
  {The observer} runs \Init{} and sends the initial memory layout
  $\servstate$ to $\adversary$;
\item For $t = 1, 2, . . . , poly(\statsec)$, the adversary
  $\adversary$ adaptively chooses an operation $op_t$ where
  $op_t$ is either \Read{}, \Write{} or \Audit{}.
  The observer executes the respective algorithms with
  $\adversary$, updating \clstate{} and \db{} according to the \Write{}
  operations specified;
\item The observer runs \(\nextractb\) \Audit{} algorithms with
  \adversary{} and records the outputs
  \(\pi_1,\ldots,\pi_{\nextract'}\) of those which did not return \reject{},
  where \(0\le \nextract' \le \nextractb\).
\item The adversary~\adversary{} is said to with the game if
  \(\nextract'\ge\nextractb /2\) and
  \(\Extract(\clstate,\pi_1,\ldots,\pi_\nextract)\ne\db\).
\end{enumerate}

\begin{definition}[Retrievability]\label{def:retr}
A PoR scheme satisfies
retrievability if no polynomial-time adversary
$\adversary$ has more than negligible probability in winning the
above security game.
\end{definition}

\section{Time-space tradeoff lower bound}
\label{sec:lowerbound}
\newcommand{\Recover}{\ttb{Recover}}

The state of the art in Proofs of Retrievability
schemes consists of some approaches with a low audit cost but a high
storage overhead (e.g., \cite{juels2007pors,Shi:2013:orampor,Cash:2017:DPR}) and some schemes with a low
storage overhead but high computational cost for the server during
audits (e.g.,
\cite{ateniese2007provable,Sebe:2008:EfficientRD,shacham2008compact}).

Before presenting our own constructions (which fall into the latter
category) we prove that there is indeed an inherent tradeoff in any PoR
scheme between the amount of extra storage and the cost of performing
audits. By \emph{extra storage} here we mean exactly the number of extra
bits of persistent memory, on the client or server, beyond the
bit-length of the original database being represented.

\Cref{thm:lowerbound} below shows that, for any PoR scheme with
sub-linear audit cost, we have
\begin{equation}
  (\text{extra storage size})
  \cdot \frac{\text{audit cost}}{\log(\text{audit cost})}
  \in \Omega(\text{data size}).
\end{equation}

None of the previous schemes, nor those which we
present, make this lower bound tight. Nonetheless, it demonstrates that
a ``best of all possible worlds'' scheme with, say, $O(\sqrt{N})$ extra
storage and $O(\log N)$ audit cost to store an arbitrary $N$-bit database,
is impossible.

The proof is by contradiction, presenting an attack on an arbitrary PoR
scheme which does not satisfy the claimed time/space lower bound. Our
attack consists of flipping $k$ randomly-chosen bits of the storage.
First we show
that $k$ is small enough so that the audit probably does not examine any
of the flipped bits, and still passes. Next we see that $k$ is large
enough so that, for some choice of the $N$ bits being represented,
flipping $k$ bits will, with high probability, make it impossible for
any algorithm to correctly recover the original data. This is a
contradiction, since the audit will pass even though the data is lost.

Readers familiar with coding theory will notice that the second part of
the proof is similar to Hamming's bound for the minimal distance of a
block code. Indeed, view the original $N$-bit data as a
\emph{message}, and the storage using $s+c$ extra bits of memory as an
$(N+s+c)$-bit \emph{codeword}: a valid PoR scheme must be able to
extract (\emph{decode}) the original message from an $(N+s+c)$-bit string,
or else should fail any audit.

\begin{restatable}{theorem}{thmlowerbound}\label{thm:lowerbound}
  For any Proof of Retrievability scheme which stores an arbitrary
  database of $N$~bits, uses at most $N+s$~bits of persistent memory on
  the server, $c$~bits of persistent memory on the client,
  and requires at most $t$~steps to perform an audit. Assuming~$s\ge{0}$, then
  either
  $t{>}\frac{N}{4}$, or
  \begin{equation}
    (s+c)\frac{t}{\log_2 t} \ge \frac{N}{12}.
    \end{equation}
\end{restatable}

\begin{proof}
  First observe that $N=0$ and $t=0$ are both trivial cases: either the
  theorem is always true, or the PoR scheme is not correct. So we assume
  always that $N\ge 1$ and $t\ge 1$.

  By way of contradiction, suppose a valid PoR scheme exists with
  $s \ge 0$, $t \le \frac{N}{4}$, and
  \begin{equation}
    (s+c)\frac{t}{\log_2 t} < \frac{N}{12}.
      \label{eqn:lbcontra}
  \end{equation}

  Following the definitions in \cref{sec:model}, we consider only the
  \Audit{} and \Extract{} algorithms.
  The \Audit{} algorithm may be randomized and, by
  our assumption, examines at most $t$ bits of the underlying memory.
  At any point in an \emph{honest} run of the algorithm, the server
  stores a \((N+s)\)-bit string \servstate{}, the client stores a
  \(c\)-bit string \clstate{}, and the \emph{client virtual memory} in
  the language of \cite{Cash:2017:DPR} is the unique \(N\)-bit string
  \db{} such that \(\IsValid(\clstate,\servstate,\db)\).

  Define a map
  \(\phi: \{0,1\}^{N+s+c} \to \{0,1\}^N\) as follows.
  Given any pair \((\clstate,\servstate)\) of length-\(N+s\) and
  length-\(c\) bit strings, run
  \(\Extract(\clstate,\Audit_1(\clstate,\servstate),\ldots,
    \Audit_\nextract(\clstate,\servstate))\)
  repeatedly over all possible choices of randomness, and record the
  majority result. By \cref{def:corr}, we have that
  \(\phi(\clstate,\servstate)=\db\) whenever
  \(\IsValid(\clstate,\servstate,\db)\).

  Observe that this map \(\phi\) must be onto, and
  consider, for any $N$-bit data string \db{},
  the preimage \(\phi^{-1}(\db)\), which is the set
  of client/server storage configurations \((\clstate,\servstate)\) such that
  \(\phi(\clstate,\servstate)=\db\).
  By a pigeon-hole argument, there must exist some
  string $\db_0$ such that
  \begin{equation}
    \#\phi^{-1}(\db_0) \le \frac{2^{N+s+c}}{2^N} = 2^{s+c}.
    \label{eqn:lbsm0}
  \end{equation}
  Informally, $\db_0$ is the data which is most easily corrupted.

  We now define an adversary \adversary{}
  for the game of \cref{def:retr} as follows:
  On the first step, \adversary{} chooses \(\db_0\) as the initial
  database, and uses this in the \Init{} algorithm to receive server
  state \servstate{}. Next, \adversary{} chooses \(k\) indices uniformly at
  random from the \servstate{} of \((N+s)\) bits
  (where \(k\) is a parameter to be defined next), and flips
  those \(k\) bits in \servstate{} to obtain a \emph{corrupted} state
  \(\servstate'\). Finally, \adversary{} runs the
  honest \Audit{} algorithm \(2\nextract\) times on step 3 of the
  security game, using this corrupted state \(\servstate'\).

  What remains is to specify how many bits \(k\) the adversary should
  randomly flip, so that most of the \(2\nextract\) runs of the \Audit{}
  algorithm succeed, but the following call to \Extract{} does not
  produce the original database \(\db_0\).
  Let
  \begin{equation}
    k = \left\lfloor \frac{N+s}{4t} \right\rfloor.
      \label{eqn:lbk}
  \end{equation}
  We assumed that $s\ge 0$ and $t \le \frac{N}{4}$, thus we have that $k\ge 1$.

  Let \(\clstate\) be the initial client state (which is unknown to
  \adversary{}) in the attack above with initial database \(\db_0\).
  From the correctness requirement (\cref{def:corr}) and the definition
  of \(t\) in our theorem, running
  \(\Audit(\clstate,\servstate)\) must always succeed after examining at
  most \(t\) bits of \servstate{}. Therefore, if
  the $k$ flipped bits in the corrupted server storage \(\servstate'\)
  are not among the (at most) $t$ bits examined by the
  \Audit{} algorithm, it will still pass. By the union bound, the
  probability that a single run of $\Audit{}(\clstate,\servstate')$ passes is at least
  \[
    1 - t\frac{k}{N+s} \ge \frac{3}{4}.
  \]
  This means that the expected number of failures in running
  \(2\nextract\) audits is \(\tfrac{\nextract}{2}\), so the Markov
  inequality tells us that the adversary \adversary{} successfully passes
  at least \nextract{} audits (as required) with probability at least
  \(\tfrac{1}{2}\).
  We want to examine the probability that
  \(\phi(\clstate,\servstate')\ne\db_0\), and therefore that the final
  call to \Extract{} in the security game does not produce \(\db_0\) and
  the adversary wins with high probability.
  Because there are \(\binom{N+s}{k}\) distinct ways to choose the \(k\)
  bits to form corrupted storage \(\servstate'\), and from the upper
  bound of \eqref{eqn:lbsm0} above, the probability that
  \(\phi(\clstate,\servstate')\ne\db_0\) is at least
  \begin{equation}
    1 - \frac{2^{s+c}-1}{\binom{N+s}{k}}.
    \label{eqn:lbprob}
  \end{equation}

  Trivially, if $s+c=0$, then this probability equals 1.
  Otherwise, from the original assumption
  \eqref{eqn:lbcontra}, and because
  $\log_2(4t)/(2t) \le 1$ for all positive integers $t$, we have
  \[
    s+c+2 \le 3(s+c) < \frac{N\log_2 t}{4t} \le
    \left(\frac{N}{4t}-1\right)\log_2(4t).
  \]

  Therefore
  \[
    \binom{N+s}{k} \ge \left(\frac{N+s}{k}\right)^k >
    (4t)^{\tfrac{N+s}{4t} - 1} \ge 2^{s+c+2}.
  \]

  Returning to the lower bound in \eqref{eqn:lbprob}, the
  probability that
  the final \Extract{} does not return \(\db_0\) is
  at least $\tfrac{3}{4}$. Combining with the first part of the proof,
  we see that, with probability at least
  $\tfrac{3}{8}$, the attacker succeeds:
  at least \nextract{} runs of \(\Audit(\clstate,\servstate')\) pass,
  but the final run of \Extract{} fails to produce the correct database
  \(\db_0\).
\end{proof}

\section{Retrievability via verifiable computing}
\label{sec:simple}
We first present a simple version of our PoR protocol.
This version contains the main ideas of our approach, namely, using
matrix-vector products during audits to prove retrievability.
It also makes use of Merkle hash trees during reads and updates to
ensure authenticity.

This protocol uses only \(N + o(N)\) persistent server storage, which is
an improvement to the \(O(N)\) persistent storage of existing PoR
schemes, and is the main contribution of this work. The costs of our
\Read{} and \Write{} algorithms are similar to existing work, but we
incur an asymptotically higher cost for the \Audit{} algorithm, namely
\(O(\sqrt{N})\) communication bandwidth and \(O(N)\) server computation
time. We demonstrate in the next section that this tradeoff between
persistent storage and \Audit{} cost is favorable in cloud computing
settings for realistic-size databases.

Later, in \cref{sec:formal}, we
give a more general protocol and prove it secure according to the PoR
definition in \cref{sec:model}. That generalized version shows how to
achieve \(O(1)\) persistent client storage with the same costs, or
alternatively
to arbitrarily decrease communication bandwidth during \Audit{}s by increasing
client persistent storage and computation time.
\subsection{Overview}

A summary of our four algorithms is shown in \cref{proto:first-scheme},
where dashed boxes are the classical, Merkle hash tree authenticated,
remote read/write operations.

Our idea is to use verifiable computing schemes as, e.g., proposed
in~\cite{Fiore:2012:PVD}.
Our choice for this is to treat the data as a square matrix of dimension
roughly \(\sqrt{N}\times\sqrt{N}\).
This allows for the matrix multiplication verification described
in \cite{Freivalds:1979:certif} to
be used as a computational method for the audit algorithm.

Crucially, this does not require any additional metadata; the database
\(M\) is stored as-is on disk, our algorithm merely treats the machine
words of this unmodified data as a matrix stored in row-major order.
Although the computational complexity for the \Audit{} algorithm is
asymptotically \(O(N)\) for the server, this entails only a single
matrix-vector multiplication, in contrast to some prior work which
requires expensive RSA computations \cite{ateniese2007provable}.

To ensure authenticity also during \Read{} and \Write{} operations, we
combine this linear algebra idea  with a standard Merkle hash tree.

\begin{figure*}[htbp]\centering
  \caption{Client/server PoR protocol with low storage server}\label{proto:first-scheme}
  \begin{tabular}{|c|ccc|}
    \hline
    & Server & Communications & Client \\
    \hline
    \multirow{4}{*}{\Init}&& $N=m n \log_2 q$ & $\uu\random{\F_q^m}$\\
    & & & $\vv\tp\leftarrow{}\uu\tp\MM$. \\
    & \multicolumn{3}{c|}{\(\begin{array}{:rcl:}
      \cdashline{1-3}
      & \multirow{2}{*}{\mtinit} & \longleftarrow \compsec,\statsec, b, \MM \\
      \MM, T_\MM \longleftarrow & & \longrightarrow r_\MM \\
      \cdashline{1-3}
    \end{array}\)} \\
    & Stores $\MM$ and $T_\MM$& & Stores \(\uu,\vv\), and \(r_\MM\) \\
    \hline
    \multirow{2}{*}{\Read}
    & \multicolumn{3}{c|}{\rule{0pt}{8 mm}\(\begin{array}{:rcl:}
      \cdashline{1-3}
      \MM, T_\MM \longrightarrow & \multirow{2}{*}{\mtread} & \longleftarrow i, j, r_\MM \\
      & & \longrightarrow \MM_{ij} \\
      \cdashline{1-3}
    \end{array}\)} \\
    & & & Returns \(\MM_{ij}\) \\
    \hline
    \multirow{3}{*}{\Write}
    & \multicolumn{3}{c|}{\rule{0pt}{8 mm}\(\begin{array}{:rcl:}
      \cdashline{1-3}
      \MM, T_\MM \longrightarrow & \multirow{2}{*}{\mtwrite} & \longleftarrow i, j, \MM_{ij}', r_\MM \\
      \MM', T_\MM' \longleftarrow & & \longrightarrow \MM_{ij}, r_\MM' \\
      \cdashline{1-3}
    \end{array}\)} \\

    & & & $\vv'_j\leftarrow{}\vv_j+\uu_i(\MM'_{ij}-\MM_{ij})$ \\
    & Stores updated $\MM', T_\MM'$ & & Stores updated \(r_\MM', \vv'\) \\

    \hline
    \multirow{2}{*}{\Audit} & & $\overset{\xx}\longleftarrow$ &
    $\xx\random{\F_q^n}$ \\
    & $\yy\gets\MM\xx$ &
    $\overset{\yy}\longrightarrow$ & $\uu\tp\yy\checks{=}\vv\tp\xx$ \\
    \hline
  \end{tabular}
\end{figure*}

\subsection{Matrix based approach for audits} \label{ssec:First-PoR}
The basic premise of our particular PoR is to treat the data, consisting
of $N$ bits , as a matrix
$\MM\in\F_q^{m\times n}$, where $\F_q$ is a suitable finite field
of size $q$,
{and each chunk of $\lfloor \log_2 q\rfloor$ bits is
  considered as an element of~$\F_q$.}
Crucially, the choice of field $\F_q$ detailed below does not require any
modification to the raw data itself; that is, any element of the matrix
$\MM$ can be retrieved in $O(1)$ time.
At a high level, our audit algorithm follows the matrix multiplication
verification technique of \cite{Freivalds:1979:certif}.

In the \Init{} algorithm,
the client chooses a secret random \emph{control vector} $\uu \in \F_q^m$
and computes a second secret control vector $\vv\in\F_q^n$ according to
\begin{equation}\vv\tp= \uu\tp\MM.\end{equation}

Note that \uu{} is held constant for the duration of the storage.
This does not compromise security because no message which depends on
\uu{} is ever sent to the Server. In particular, this means that
multiple clients could use different, independent, control vectors \(\uu\)
as long as they have a way to synchronize \Write{} operations
(modifications of their shared database) over a secure channel.

To perform an audit, the client chooses a random \emph{challenge vector}
$\xx\in\F_q^n$, and asks the server to compute a \emph{response vector}
$\yy\in\F_q^m$ according to
\begin{equation}
\yy=\MM\xx
\end{equation}
Upon receiving the response $\yy$, the client checks two dot
products for equality, namely
\begin{equation}\label{dotprod}
\uu\tp\yy \stackrel{?}{=} \vv\tp\xx.
\end{equation}
The proof of retrievability will rely on the fact that observing several
successful audits allows, with high probability, recovery of the
matrix $\MM$, and therefore of the entire database.

The audit algorithm's cost is mostly in the server's matrix-vector product.
The client's dot products are much cheaper in comparison.
For instance if $m=n$ are close to $\sqrt{N}$,
the communication cost is bounded by $O(\sqrt{N})$  as each vector
has about $\sqrt{N}$ values.
We trade this infrequent heavy computation for {almost} no additional
persistent storage {on the server side}, justified by the significantly cheaper cost of
computation versus storage space.

A sketch of the security proofs is as follows; full proofs are provided
along with our formal and general protocol in \cref{sec:formal}.
The Client knows that the Server sent the correct value of \(\yy\) with
high probability, because otherwise the Server must know something
about the secret control vector \(\uu\) chosen randomly at
initialization time. This
is impossible since no data depending on \(\uu\) was ever sent to the
Server. The retrievability property (\cref{def:retr}) is ensured from the fact that, after
\(\sqrt{N}\) random successful audits, with high probability, the
original data \(\MM\) is the unique solution to the matrix equation
\(\MM\XX=\YY\), where \(\XX\) is the matrix of random challenge vectors
in the audits and \(\YY\) is the
matrix of corresponding response vectors from the Server.

Some similar ideas were used by \cite{Sebe:2008:EfficientRD} for
checking integrity.
However, their security relies on the difficulty of
integer factorization.
Implementation would therefore require
many modular exponentiations at thousands of bits of precision. Our
approach for audits is much simpler and independent of computational hardness
assumptions.

\subsection{Merkle hash tree for updates}\label{ssec:Merkle}
In our protocols, the raw database of size $N$ bits is handled in two different ways.
As seen in the previous section, the audits use chunks of $\lfloor \log_2 q\rfloor$ bits as elements
of a finite field~$\F_q$.
Second, a Merkle hash tree with a different block size \mtblock{}
is used here to
ensure authenticity of individual \Read{} operations.
This tree is a binary tree, stored on the server, consisting
of \(O(N/\mtblock)\) hashes, each of size
\(2\compsec\), for collision resistance.

The Client stores only the root
hash, and can perform, with high integrity assurance,
any read or write operation on a range of \(k\) bytes
in \(O(k + \mtblock + \log(N/\mtblock))\) communication and computation
time.
When the block size is large enough, the extra server storage is
\(o(N)\); for example, \(\mtblock\ge\log N\) means the hash tree can be
stored using \(O(N\compsec/\log N)\) bits.

Merkle hash trees are a classical result, commonly used in practice, and we do
not claim any novelty in our use here \cite{Merkle,rfc6962}.
To that
end, we provide three algorithms to abstract the details of the Merkle
hash tree: \mtinit, \mtread{} and \mtwrite.

These three algorithms are in fact two-party protocols between a Server and a
Client, but without any requirement for secrecy. A vertical bar~\(\mid\) in the
inputs and/or outputs of an algorithm indicates Server input/output on the
left, and Client input/output on the right. When only the Client has
input/output, the bar is omitted for brevity.

The \mtread{} and \mtwrite{} algorithms may both fail to verify a hash,
and if so, the Client outputs \reject{} and aborts immediately.
Our three Merkle tree algorithms are as follows.

\par{\(\mtinit(1^\compsec, \mtblock, \db) \mapsto (\db, T_\db \mid r_\db)\).}
    The Client initializes database \(\db\) for storage in
    size-\mtblock{} blocks. The entire database \db{} is sent to the
    Server, who computes hashes and stores the resulting Merkle hash
    tree \(T_\db\). The Client also computes this tree, but discards all
    hashes other than the root hash \(r_\db\).
    The cost in communication and computation for both parties is bounded by
    \(O(|\db|)=O(N)\).
\par{\(\mtread(\db, T_\db \mid range, r_\db) \mapsto \db_{range}\).}
    The Client sends a contiguous byte range to the server, i.e., a pair
    of indices within the size of \db{}. This range determines which
    containing range of blocks are required, and sends back these block
    contents, along with left and right boundary paths in the hash tree
    \(T_\db\). Specifically, the boundary paths include all left sibling
    hashes along the path from the first block to the root node, and all
    right sibling hashes along the path from the last block to the root;
    these are called the ``uncles'' in the hash tree.
    Using the returned blocks and hash tree values, the
    Client reconstructs the Merkle tree root, and compares with
    \(r_\db\). If these do not match, the Client outputs \reject{}
    and aborts. Otherwise, the requested range of bytes is extracted
    from the (now-verified) blocks and returned.
    The cost in communication and computation time for both
    parties is at most \(O(|range| + \mtblock + \log(N/\mtblock))\).
\par{\(\mtwrite(\db, T_\db \mid range, \db_{range}', r_\db)\)}\,\\
\hspace*{1em}\hfill\(\mapsto (\db', T_\db' \mid \db_{range}, r_\db')\).\\
    The Client wishes to update the data \(\db_{range}'\) in the
    specified range, and receive the \emph{previous value} of that
    range, \(\db_{range}\), as well as an updated root hash \(r_\db\).
    The algorithm begins as \mtread{} with the Server sending all blocks
    to cover the range and corresponding left and right boundary hashes
    from \(T_\db\). After the Client retrieves and verifies the old
    value \(\db_{range}\) with the old root hash \(r_\db\), she updates
    the blocks with the new value \(\db_{range}'\) and uses the same
    boundary hashes to compute the new root hash \(r_\db'\).
    Separately, the Server updates the underlying database \(\db'\) in
    the specified range, then recomputes all affected hashes in
    \(T_\db'\).
    The asymptotic cost is identical to that for the \mtread{} algorithm.

\section{Formalization and Security analysis}\label{sec:formal}
In this section we present our PoR protocol in most general form{;} prove
it satisfies the definitions of PoR correctness, authenticity, and
retrievability{;} analyze its asymptotic performance and present a
variant that also satisfies public verifiability.

Recall that our security definition and protocol rely on two security
parameters: \compsec{} for computational security and \statsec{} for
statistical security. In our main protocol, the only dependence on
computational assumptions comes from the use of Merkle trees and the
hardness of finding hash collisions. The \compsec{} parameter will also
arise when we use encryption to extend the protocol for externalized
storage and public verifiability.

Instead, the security of our main construction mostly depends on the
statistical security parameter \statsec{}. Roughly speaking, this is
because in order to produce an incorrect result that the client will
accept for an audit, the adversary must provably \emph{guess} a result
and try it within the \emph{online} audit protocol; even observing correct
audits does not help the adversary gain an advantage. This intuition,
rigorously analyzed below, allows us to instantiate our protocol more
efficiently while providing strong security guarantees.

\subsection{Improvements on the control vectors}\label{ssec:controlvectors}

The control vectors \uu{} and \vv{} stored by the Client in the
simplified protocol from \cref{sec:simple} can be modified to
increase security and decrease persistent storage or communications.

\paragraph{Security assumptions via multiple checks.}

In order to reach a target bound $2^{-\statsec}$ on the probability of failure for authenticity,
it might \emph{theoretically} be necessary to choose multiple independent \uu{} vectors during initialization
and repeat the audit checks with each one. We will show that in fact
only one vector is necessary for reasonable settings of \statsec{}, but
perform the full analysis here for completeness
and to support a potential evolution of the security requirements.

We model multiple vectors by inflating the vectors $\uu$ and
$\vv$ to be blocks of $t$ non-zero vectors instead; that is, matrices \UU{} and
\VV{} with \(t\) rows each.
To see how large \(t\) needs to be, consider the probability of the Client
accepting an incorrect response during an audit. An incorrect answer $\zz$
to the audit fails to be detected only if
\begin{equation}\label{eq:auth:failure}
  \UU\cdot(\zz-\yy) = \vect{0},
\end{equation}
where $\yy=\MM\xx$ is the correct response which would be returned by
an honest Server, for $\MM\in\F_q^{m{\times}n}$.

If $\UU$ is sampled uniformly at random among matrices in $\F_q^{t
  \times m}$ with non-zero rows, then since the Server
never learns any information about \UU{}, {the} audit fails only if
$(\zz-\yy)\neq{\vect{0}}$ but \UU{} is in its left nullspace.
This happens with probability at most~$1/q^{t}$.

Achieving a probability bounded by $2^{-\statsec}$, requires to set
$t=\left\lceil\frac{\statsec}{\log_2(q)}\right\rceil$.
In practice, reasonable values of \(\statsec=40\) and \(q> 2^{64}\) mean
that \(t=1\) is large enough. If an even higher level of security such
as $\statsec=80$ is required, then still only 2 vectors are needed.

\paragraph{Random geometric progression.}

Instead of using uniformly random vectors $\xx$ and matrices $\UU$, one can impose a structure on them, in order
to reduce the amount of randomness needed, and the cost of communicating or storing them.
We propose to apply Kimbrel and Sinha's
modification of Freivalds' check~\cite{Kimbrel:1993:PAV}:
select a single random field element $\rho$ and form
$\xx\tp=[\rho,\ldots,\rho^{n}]$, thus reducing the communication volume for an audit from $m+n$ to
$m+1$ field elements.

Similarly, we can reduce the storage of $\UU$
by sampling uniformly at random \(t\)
distinct non-zero elements \(s_1,\dots,s_t\) and forming
\begin{equation}\label{eq:ufroms}
\UU=\begin{smatrix}
  s_1 & \cdots &s_1^m \\
  \vdots && \vdots \\
  s_t & \cdots & s_t^m
\end{smatrix} \in \F_q^{t \times m}.
\end{equation}
This reduces the storage on the client side from \(mt + n\) to only \(t + n\) field elements.

Then with a rectangular database and $n>m$,
communications can be potentially lowered to any small target amount,
at the cost of increased client storage and greater client computation during audits.
This structure constraint on $\UU$ impacts the probability of failure
of the authenticity for the audits.
Consider an incorrect answer \zz{} to an audit as in \eqref{eq:auth:failure}.
Then each element \(s_1,\ldots,s_t\) is a root of the
degree-$(m-1)$ univariate polynomial whose coefficients are $\zz-\yy$.
Because this polynomial has at most \(m-1\) distinct roots, the probability
of the Client accepting an incorrect answer is at most
\begin{equation}
\frac{\binom{m-1}{t}}{\binom{q}{t}} \le \left(\frac{m}{q}\right)^t,
\end{equation}
which leads to setting
$t = \left\lceil\frac{\statsec}{\log_2(q)-\log_2(m)}\right\rceil$ in order to bound this probability
by $2^{-\statsec}$. Even if \(N=2^{53}\) for 1PB of storage, assuming \(m\le n\),
and again using \(\statsec=40\) and \(q\ge 2^{64}\), still $t=1$
suffices.

\paragraph{Externalized storage.}

Lastly, the client storage can be reduced to \(O(\compsec)\) by externalizing the storage of the block-vector \VV{}
at the expense of increasing the volume of communication.
Clearly \VV{} must be stored
encrypted, as otherwise the server could answer any challenge without having to store the database.
Any IND-CPA symmetric cipher works here, with care taken so that a separate IV is used for each column;
this allows updates to a column of \VV{} during a \Write{} operation without revealing anything about
the updated values.

In the following we will thus simply assume that the client has access
to an encryption function $E_K{:\F_q\rightarrow\CC}$
({from the field to any ciphertext space $\CC$}) and a decryption
function $D_K{:\CC\rightarrow\F_q}$, both parameterized with a secret key $K$.
In order to assess the authenticity of each communication of the
ciphered $\VV$ from the Server to the client, we will use another
Merkle-Hash tree certificate for it: the client will only need to keep
the root of a Merkle-Tree built on the encryption of $\VV$.
With this, we next show how to efficiently and securely update both
the database and this externalized ciphered control vector. Further,
this ensures non-malleability outside of the encryption scheme:
INT-CTXT (integrity of ciphertexts) together with IND-CPA implies
IND-CCA2~\cite[Theorem~2]{Bellare:2000:AuthEncrypt}.

Since this modification reduces the client storage but increases the
overall communication, we consider
both options (with or without it; \ExternT{} or \ExternF{}), and we
state the algorithms for our protocol with a \textit{Strategy}
parameter, deciding whether or not to externalize the storage
of~$\VV$.

\begin{table*}[htbp]\centering
\caption{Proof of retrievability via rectangular verifiable computing with structured vectors}\label{tab:porvcboth}
{\normalfont\footnotesize
$N=mn\log_2{q}$ is the size of the database, $\compsec\geq\statsec$
are the computational and statistical security parameters,
$\mtblock{}>\compsec\log N$ is the Merkle tree block size.\\\vspace{-2pt}
Assume \(\log_2 q\) is a constant.}\\

\begin{tabular}{|c|c||c||cc||cc|}
\hline
\multicolumn{2}{|c||}{}& Server & \multicolumn{2}{c||}{Communication}
& \multicolumn{2}{c|}{Client} \\
\hline
\multicolumn{2}{|c||}{Strategy} & & \ExternT & \ExternF & \ExternT & \ExternF\\
\hline
\hline
\multicolumn{2}{|c||}{Storage} & $N+ O(N \compsec/\mtblock)$ &
\multicolumn{2}{c||}{} & $O(\compsec)$& $O\left( n \compsec \right)$ \\
\hline
\hline
\multirow{3}{*}{\rotatebox[origin=c]{90}{Comput.}} &Setup & $O(N)$ &
$N+o(N)$ & $N$ & \multicolumn{2}{c|}{$O(N)$}\\
&Audit & $N$ & $O(m+n\compsec)$ & $O(m)$ & \multicolumn{2}{c|}{$O(\compsec(m+n))$} \\
&Read/Write & $O(\mtblock + \compsec \log N)$ &
\multicolumn{2}{c||}{$O(\mtblock + \compsec \log N )$} &
\multicolumn{2}{c|}{$O\left(\mtblock + \compsec \log N \right)$} \\
\hline
\end{tabular}
\end{table*}

\subsection{Formal protocol descriptions}

Full definitions of the five algorithms, \Init{}, \Read{}, \Write{},
\Audit{} and \Extract{}, as
\cref{alg:init,alg:read,alg:write,alg:audit,alg:extract}, are given
below, incorporating the improvements on control vector storage from the
previous subsection.
They include subcalls to the classical Merkle hash tree operations
defined in Section~\ref{ssec:Merkle}.

Then, a summary of the asymptotic costs can be found in
\cref{tab:porvcboth}.

\begin{algorithm}[htb]
  \caption{$\Init(1^\compsec,1^\statsec,m,n,q,\mtblock,M,Strategy)$}\label{alg:init}
  \begin{algorithmic}[1]
 \REQUIRE $1^\compsec,1^\statsec; m,n,q,\mtblock\in\N; \MM \in \F_q^{m\times n}$
 \ENSURE $\servstate{}$, $\clstate{}$

\STATE $t \leftarrow \lceil\statsec/(\log_2 q)\rceil \in \N$;
    \STATE Client: $\svec\random{\F_q^t}$ with non-zero distinct elements\hfill\COMMENT{Secrets}
    \STATE Client: Let \(\UU \gets [\svec_i^j]_{i=1\ldots{}t,j=1\ldots{}m} \in \F_q^{t\times{}m}\)
    \STATE Client: \(\VV \gets \UU\MM \in \F_q^{t\times{}n}\)
      \hfill\COMMENT{Secretly stored or externalized}
    \STATE Both: \((\MM, T_\MM \mid r_\MM) \gets \mtinit(1^\compsec,
    \mtblock, \MM)\)

\IF{$ (Strategy = externalization)$}
 \STATE Client: $K\random{\K}$;
\STATE  Client: $\WW \leftarrow E_K(\VV) \in \CC^{t\times
  n}$;\hfill{\COMMENT{elementwise}}
\STATE Client: sends $m,n,q,\MM,\WW$ to the Server;
\STATE Both: \((\WW, T_\WW \mid r_\WW) \gets \mtinit(1^\compsec, \mtblock, \WW)\)
\STATE Server: $\servstate{}  \leftarrow (m,n,q,\MM,T_\MM,Strategy,\WW,T_\WW)$;
\STATE Client: $\clstate{} \leftarrow (m,n,q,t,\svec, r_\MM,Strategy,K,r_\WW)$;
\ELSE
\STATE Client: sends $m,n,q,\MM$ to the Server;
\STATE Server: $\servstate{}  \leftarrow (m,n,q,\MM,T_\MM,Strategy)$;
\STATE  Client: $\clstate{} \leftarrow (m,n,q,t,\svec, r_\MM,Strategy,\VV)$;
\ENDIF
  \end{algorithmic}
\end{algorithm}

\begin{algorithm}[htb]
  \caption{$\Read(\servstate{},\clstate{}, i,j)$}\label{alg:read}
  \begin{algorithmic}[1]
    \REQUIRE  $\servstate{},$\clstate{}$, i \in [1..m] ,j\in [1..n]$
    \ENSURE $\MM_{ij}$ %
      or \reject
\STATE Both: \(\MM_{ij} \gets \mtread(\MM, T_\MM \mid (i,j), r_\MM)\)
\CRETURN{Client} $\MM_{ij}$
  \end{algorithmic}
\end{algorithm}

\begin{algorithm}[htb]
  \caption{$\Write(\servstate{},\clstate{}, i,j,\MM'_{ij},Strategy)$}\label{alg:write}
  \begin{algorithmic}[1]
\REQUIRE $\servstate{}, \clstate{}, i \in [1..m] ,j\in [1..n], \MM'_{ij} \in \F_q$
\ENSURE $\servstate', \clstate'$ or  \reject
\STATE Both: \((\MM', T_\MM' \mid \MM_{ij}, r_\MM')\)\\
  \hfill \(\gets \mtwrite(\MM, T_\MM \mid (i,j), \MM'_{ij}, r_\MM)\)

\IF{$(Strategy=externalization)$}
\STATE Both: \(\WW_{1..t,j} \gets \mtread(\WW, T_\WW \mid (1..t,j),r_\WW)\)
\STATE Client: $\VV_{1..t,j} \leftarrow D_K(\WW_{1..t,j}) \in \F_q^{ t}$;
\ENDIF
\STATE Client: Let \(\UU_{1..t,i} \gets [\svec_k^i]_{k=1\ldots{}t} \in \F_q^t\)
\STATE Client: $  \VV'_{1..t,j} \leftarrow \VV_{1..t,j} + \UU_{1..t,i}(\MM'_{ij}-\MM_{ij})  \in \F_q^{ t} $;
\IF{$(Strategy=externalization)$}
\STATE Client:  $\WW'_{1..t,j} \leftarrow E_K(\VV'_{1..t,j} ) \in \CC^{ t}$
\STATE Both: \((\WW', T_\WW' \mid \WW_{1..t,j},r'_\WW)\)\\
  \hfill\(\gets \mtwrite(\WW, T_\WW \mid (1..t,j), \WW'_{1..t,j}, r_\WW)\)
\STATE Server: Update $\servstate'$ using \(\MM',T_\MM',\WW'\), and \(T_\WW'\)
\STATE Client: Update $\clstate'$ using \(r_\MM'\) and \(r'_\WW\)
\ELSE
\STATE Server: Update $\servstate'$ using \(\MM'\) and \(T_\MM'\)
\STATE Client: Update $\clstate'$ using \(r_\MM'\) and \(\VV'\)
\ENDIF
  \end{algorithmic}
\end{algorithm}

\begin{algorithm}[htb]
  \caption{$\Audit(\servstate{},\clstate{},Strategy)$}\label{alg:audit}
  \begin{algorithmic}[1]
    \REQUIRE $\servstate{}$, $\clstate{}$
    \ENSURE  \accept{} or \reject
    \STATE Client: $ \rho \random{\F_q}$ and sends it to the Server;
    \STATE Let $\xx\tp\leftarrow[\rho^1,\rho^2, \ldots, \rho^{n}]$
\STATE Server:
$\yy\leftarrow{}\MM\xx \in \F_q^{m}$;\hfill\COMMENT{\MM{} from $\servstate{}$}
\STATE Server: sends $\yy$ to Client;
\IF{$ (Strategy = externalization)$}
\STATE Both: \(\WW \gets \mtread(\WW, T_\WW \mid (1..t,1..n), r_\WW)\);
\STATE Client: $\VV\leftarrow D_K(\WW) \in \F_q^{t\times{}n} $
\ENDIF
\STATE Client: Let \(\UU \gets [\svec_i^j]_{i=1\ldots{}t,j=1\ldots{}m} \in \F_q^{t\times{}m}\)
\IF{\((\UU\yy = \VV\xx)\)}
\CRETURN{Client} \accept
\ELSE
\CRETURN{Client} \reject
\ENDIF
  \end{algorithmic}
\end{algorithm}

\begin{algorithm}[htb]
  \caption{\(\Extract(\clstate{},(\xx_1,\yy_1),\ldots,(\xx_\nextract,\yy_\nextract))\)%
    \label{alg:extract}}
  \begin{algorithmic}[1]
   \REQUIRE \(\clstate{}\) and \(\nextractb \ge 4n + 24\statsec\) audit
      transcripts \((\xx_i,\yy_i)\), of which more than $e/2$ are successful.
    \ENSURE \MM{} or \fail
    \STATE \(\ell_1,\ldots,\ell_k \gets \) indices of \emph{distinct}
      successful challenge vectors \(\xx_{\ell_i}\)
    \IF{\(k < n\)}
      \RETURN \fail
    \ENDIF\hfill\COMMENT{Now $\XX$ is Vandermonde with distinct points}
    \STATE Form matrix \(\XX \gets [\xx_{\ell_1} | \cdots | \xx_{\ell_n}] \in \F_q^{n\times{}n}\)
    \STATE Form matrix \(\YY \gets [\yy{\ell_1} | \cdots | \yy{\ell_n}] \in \F_q^{m\times{}n}\)
    \STATE Compute \(\MM \gets \YY\,\XX^{-1}\)
    \RETURN \(\MM\)
  \end{algorithmic}
\end{algorithm}

\subsection{Security}

Before we begin the full security proof, we need the following technical
lemma to prove that the \Extract{} algorithm succeeds with high probability.
The proof of this lemma is a straightforward application of Chernoff
bounds.

\begin{restatable}{lemma}{chernoff}
\label{lem:chernoff}
  Let \(\statsec,n\ge 1\) and
  suppose \(\nextractb \) balls are thrown independently and uniformly into
  \(q\) bins at random. If \(\nextractb=4n+24\statsec\) and
  \(q\ge{}4\nextractb \), then with probability at least
  \(\exp(-\statsec)\), the number of non-empty bins is at
  least \(\nextractb /2 + n\).
\end{restatable}
\begin{proof}
Let \(B_1,B_2,\ldots,B_{\nextractb}\) be random variables for
the indices of bins that each ball goes into. Each is a uniform independent
over the \(q\) bins.
Let \(X_{1,2},X_{1,3},\ldots,X_{\nextractb-1,\nextractb}\) be
\(\binom{\nextractb}{2}\) random variables for
each pair of indices \(i,j\) with \(i\ne j\), such that \(X_{i,j}\) equals 1
iff \(B_i = B_j\). Each \(X_{i,j}\) is a therefore Bernoulli trial with
\(\E[X_{i,j}] = \tfrac{1}{q}\), and the sum
\(X=\sum_{i\ne j} X_{i,j}\) is the number of pairs of balls which go into the same bin.

We will use a Chernoff bound on the probability that \(X\) is large.
Note that the random variables \(X_{i,j}\) are \emph{not} independent, but
they are negatively correlated: when any \(X_{i,j}\) equals 1, it only
\emph{decreases} the conditional expectation of any other \(X_{i',j'}\).
Therefore, by convexity, we can treat the \(X_{i,j}\)'s as independent
in order to obtain an upper bound on the probability that \(X\) is large.

Observe that \(\E[X] = \binom{\nextractb}{2} / q < \nextractb/8\). A standard consequence
of the Chernoff bound on sums of independent indicator variables tells us that
\(\Pr[X \ge 2\E[X]] \le \exp(-\E[X]/3)\); see for example
\cite[Theorem 4.1]{MRbook},
or \cite[Theorem 1]{harveynotes}.

Substituting the bound on \(\E[X]\) then tells us that
\(\Pr[X \ge \nextractb/4] \le \exp(-\nextractb/24) < \exp(-\statsec)\).
That is, with high probability, fewer than \(\nextractb/4\) pair of balls share the
same bin. If $n_k$ denotes the number of bins with $k$ balls, the number of non-empty bins is:
\[\begin{split}
\sum_{k=1}^q n_k = \left(\nextractb -  \sum_{k=2}^q k n_k\right) + \sum_{k=2}^q n_k &= \nextractb -\sum_{k=2}^q
(k-1)n_k\\
&\geq \nextractb- \sum_{k=2}^q
\binom{k}{2} n_k. %
\end{split}\]
The latter is $> \frac{3}{4} \nextractb$ with high probability, which
completes the proof, since \(3\nextractb/4=\nextractb/2+\nextractb/4=\nextractb/2+n+6\statsec\).
\end{proof}

We now proceed to the main result of the paper.

\begin{restatable}{theorem}{thmaudit}\label{thm:audit}
Let \(\compsec,\statsec,m,n\in\N\), \(\F_q\) a finite field satisfying
\(q\ge 16n + 96\statsec\) be parameters for our PoR scheme.
Then the protocol composed of:
\begin{itemize}[nosep]
\item the \Init{} operations in~\cref{alg:init};
\item the \Read{} operations in~\cref{alg:read};
\item the \Write{} operations in~\cref{alg:write};
\item the \Audit{} operations in~\cref{alg:audit}; and
\item the \Extract{} operation in~\cref{alg:extract} with
  \(\nextractb{=}4n{+}24\statsec\)
\end{itemize}
satisfies correctness, adaptive authenticity and retrievability
as defined in \cref{def:corr,def:privauth,def:retr}.
\end{restatable}

\begin{proof}
Correctness comes from the correctness of the Merkle hash tree
algorithms, and from the fact that, when all parties are honest,
\(\UU\yy = \UU\MM\xx = \VV\xx\).

For authenticity, first consider
the secret control block vectors $\UU$ and $\VV$.
On the one hand, in the local storage strategy, $\UU$ and $\VV$ never
travel and all the communications by the Client in all the algorithms
are independent of these secrets.
On the other hand, in the externalization strategy, $\UU$ never travels
and $\VV$ is kept confidential by the IND-CPA symmetric encryption
scheme with key \(K\) known only by the client.
Therefore, from the point of view of the server, it is equivalent,
\emph{in both strategies}, to consider either that these secrets are
computed during initialization as stated,
or that they are only determined \emph{after} the completion of any of
the operations.

Now suppose that the server sends an incorrect audit response \(\zz\ne \MM\xx\)
which the Client fails to reject, and let
\(f\in\F_q[X]\) be the polynomial with degree at most \(m-1\) whose
coefficients are the entries of \((\zz-\MM\xx)\).
Then from \eqref{eq:auth:failure} and \eqref{eq:ufroms} in the prior
discussion, each of
the randomly-chosen values \(s_1,\ldots,s_t\) is a root of this
polynomial \(f\). Because \(f\) has at most \(m-1\) distinct roots, the
chance that a single \(s_i\) is a root of \(f\) is at most \((m-1)/q\),
and therefore the probability that all \(f(s_1)=\cdots=f(s_t)=0\),
is at most \((m/q)^t\).

From the choice of \(t=\lceil\statsec/\log_2(q/m)\rceil\), the chance
that the Client fails to reject an incorrect audit response is at most
\(2^{-\statsec}\), which completes the proof of authenticity
(\cref{def:privauth}).

For retrievability, we need to prove that \cref{alg:extract} succeeds
with high probability on the last step of the security game from
\cref{def:retr}. Because of the authenticity argument above, all successful audit
transcripts are valid with probability \(1-\negl(\statsec)\); that is,
each \(\yy = \MM\xx\) in the input to \cref{alg:extract}.
This \Extract{} algorithm can find an invertible
Vandermonde matrix \(\XX \in \F_q^{n\times{}n}\), and thereby recover
\MM{} successfully, whenever at least \(n\)
of the values \(\rho\) from challenge vectors \(\xx\) are distinct.

Therefore the security game becomes essentially this: The experiment
runs the honest \Audit{} algorithm \(\nextractb = 4n+24\statsec\) times,
each time choosing a value \(\rho\) for the challenge uniformly at
random from \(\F_q\). The adversary must then select \(\nextractb /2 \) of
these audits to succeed, and the adversary wins the game by selecting
$\nextractb/2$ of the \( \nextractb\) random audit challenges which contain
fewer than \(n\) distinct \(\rho\) values.

This is equivalent to the
balls-and-bins game of \cref{lem:chernoff}, which shows that the
\Extract{} algorithm succeeds with probability at least
\(1-\exp(-\statsec) > 1-2^{-\statsec}\) for any selection of \(\nextractb / 2\) out of
\(\nextractb\) random audits.
\end{proof}

\newcommand{\G}{\ensuremath{\mathbb{G}}}
\newcommand{\Q}{\ensuremath{\mathbb{Q}}}
\newcommand{\mleft}{m}
\newcommand{\nright}{n}
\newcommand{\PK}{\matr{K}}
\newcommand{\ww}{\vect{w}}
\subsection{Publicly verifiable variant}\label{ssec:publicverif}
Our scheme can also be adapted to meet the stricter requirement of
\emph{public verifiability} (see \cref{ssec:authenticity}),
wherein there are now two types of client: a \emph{Writer} who
can run any of the \Init{}, \Write{},
\Read{}, or \Audit{} algorithms; and a \emph{Verifier} that can only run
the last two.

{The idea is that $U$ and $V$ will be published as $g^U$ and
  $g^V$ as to hide their values, while still enabling the dot product
  verification.}

{More precisely}, we will employ the externalized storage strategy outlined
in \cref{ssec:controlvectors}, so that the server holds all the
information needed to perform audits. But this alone is not enough, as the
public Verifier (and thus, the possibly-malicious server) must not learn
the plaintext control vector values.

The challenge, then, is to support equality testing for dot products
of encrypted values, without decrypting. Any linearly homomorphic
encryption could be used, but this is actually more than what we need
since decryption is not even necessary. Instead, we will simply employ a group
where the discrete logarithm is hard, instead of the (relatively
small) finite fields used before.

Further, we use a group of prime order, in order to be able to easily
compute with exponents.
In particular{, thanks to the homomorphic property
of exponentiation,}
 we will perform some linear algebra over the group and
need some notations for this.
For a matrix $\AAA$, $g^\AAA$ denotes the coefficient-wise
exponentiation of a generator $g$ to each entry in $\AAA$.
Similarly, for a matrix $\WW$ of group elements and a matrix $\BB$ of
scalars, $\WW^\BB$ denotes the extension of matrix multiplication using
the group action. If we have $\WW=g^\AAA$, then $\WW^\BB =
(g^\AAA)^\BB$.
{Futher, this quantity} can actually be computed {by working}
in the exponents first,
i.e., it is equal to $g^{(\AAA\BB)}$. For example:
\[
  \left(g^{\begin{smatrix}a&b\\c&d\end{smatrix}}\right)^{\begin{smatrix}e\\f\end{smatrix}}
  =
  \begin{smatrix}g^a&g^b\\g^c&g^d\end{smatrix}^{\begin{smatrix}e\\f\end{smatrix}}
  =
  \begin{smatrix}g^{ae+bf}\\g^{ce+df}\end{smatrix}
  =
  g^{\left(\begin{smatrix}a&b\\c&d\end{smatrix}\begin{smatrix}e\\f\end{smatrix}\right)}.
\]
The resulting modified protocol is presented formally in
\cref{proto:abspubverif}. In summary, the changes are as follows:
 \begin{enumerate}[nosep]
 \item Build a group~$\G$ of large prime order~$p$ and generator~$g$.
 \item \Init{}, in \cref{alg:init}, is run identically, except for
   two modifications: first,
   $\WW$ is {mapped to}~$\G$:
   $\WW \leftarrow E(\VV) = g^\VV$;
   second, the Writer also publishes
   an encryption of~$\UU$ as: $\PK\leftarrow g^\UU$ over an
   {\em authenticated} channel;
     $\PK$ is called the \emph{public key}.
 \item All the verifications of the Merkle tree root in
   \cref{alg:read,alg:write,alg:audit} remain unchanged,
   but the Writer must publish the new roots of the trees after each
   \Write{} also over an authenticated and timestamped channel to
   the Verifiers.
 \item Updates to the control vector, in \cref{alg:write} are performed
   homomorphically, without {``deciphering''}~$\WW$:
   the Writer computes in clear,
   $\Delta\leftarrow(\MM'_{ij}-\MM_{ij})\UU_{1..t,i}$, then updates
   $\WW'_{1..t,j}\leftarrow\WW_{1..t,j}\cdot{g^\Delta}$.
 \item The dotproduct verification, in \cref{alg:audit} is performed
   also homomorphically:
   $\PK^\yy\checks{=}\WW^\xx$.
 \end{enumerate}

\begin{remark}\label{rq:fieldsize}
  Note that the costly server-side computation during audits does not
  involve any group operations; only the clients must perform
  exponentiations. However, the field size~$p$ must be increased in
  order for the discrete logarithm to be hard.
  For a database of fixed bit-length, this increase in field size
{induces a cost overhead in the field arithmetic (up to quadratic in $\log p$) which is partly}
  compensated by a corresponding decrease in the matrix dimension {(linear in $\log p$), as $N=mn\log p$.}
\end{remark}

\begin{figure*}[htb]\centering
  \caption{Publicly verifiable Client/server PoR protocol with low storage server}\label{proto:abspubverif}
  \begin{tabular}{|c|ccc|}
    \hline
    & Server & Communications & Client \\
    \hline
    \multirow{4}{*}{\Init}&& $N=m n \log_2 q$ & $s\random{S\subseteq\Z_p}$ \\
    & & $\G$ of order $p$ and gen. $g$ & form \(\uu\gets [\svec^j]_{j=1\ldots{}m} \in \Z_p^{m}\) \\
    & & & $\vv\tp\leftarrow{}\uu\tp\MM$, $\ww\tp\leftarrow{g^\vv}\in\G^n$. \\
    & \multicolumn{3}{c|}{\(\begin{array}{:rcl:}
      \cdashline{1-3}
      & \multirow{2}{*}{\mtinit} & \longleftarrow \compsec,\statsec, b, \MM,\ww \\
      \MM, T_\MM, \ww, T_\ww \longleftarrow & & \longrightarrow r_\MM,
      r_\ww\\
      \cdashline{1-3}
    \end{array}\)} \\
    & Store $\MM,T_\MM,\ww,T_\ww$& & Publish $r_\MM$, $r_\ww$ and $\PK=g^{\uu}$\\
    \hline
    \multirow{2}{*}{\Read}
    & \multicolumn{3}{c|}{\rule{0pt}{8 mm}\(\begin{array}{:rcl:}
      \cdashline{1-3}
      \MM, T_\MM \longrightarrow & \multirow{2}{*}{\mtread} & \longleftarrow i, j, r_\MM \\
      & & \longrightarrow \MM_{ij} \\
      \cdashline{1-3}
    \end{array}\) \quad\quad Return \(\MM_{ij}\)} \\[12pt]
    \hline
    \multirow{3}{*}{\Write}
    & \multicolumn{3}{c|}{\rule{0pt}{8 mm}\(\begin{array}{:rcl:}
      \cdashline{1-3}
      \MM, T_\MM,\ww, T_\ww \longrightarrow & \multirow{2}{*}{\mtread} & \longleftarrow i, j, r_\MM, r_\ww \\
      & & \longrightarrow \MM_{ij},\ww_{j} \\
      \cdashline{1-3}
    \end{array}\)} \\
    & & & $\delta\leftarrow\uu_i(\MM'_{ij}-\MM_{ij})$ \\
    & &$\overset{i,j,\MM'_{ij},\ww'_j}\longleftarrow$  & $\ww_j'\leftarrow{}\ww_j \cdot{} g^\delta$ \\
    & Update $\MM', T_\MM',\ww',T_\ww'$ & & Publish \(r'_\MM, r'_\ww\) \\

    \hline
    \multirow{2}{*}{\Audit} &
    & $\overset{r}\longleftarrow$ &
    $r\random{S\subseteq\Z_p^*}$ \\
& $\yy\gets\MM\xx$ & form \(\xx\gets [r^i]_{i=1\ldots{}n} \in \Z_p^{n}\) & \\
    & \multicolumn{3}{c|}{\rule{0pt}{8 mm}\(\begin{array}{:rcl:}
      \cdashline{1-3}
      \ww, T_\ww \longrightarrow & \multirow{2}{*}{\mtread} & \longleftarrow r_\ww \\
      & & \longrightarrow \ww \\
      \cdashline{1-3}
    \end{array}\) \quad\quad $\PK^y \checks{=} \ww^x$ }\\[12pt]
    \hline
  \end{tabular}
\end{figure*}

Under Linearly Independent Polynomial (LIP)
Security~\cite[Theorem~1]{Abdalla:2015:PRF}\footnote{LIP security
  reduces to the MDDH hypothesis, a generalization of the widely used
  decision linear assumption~\cite{Abdalla:2015:PRF,Morillo:2016:KMDDH}},
the Protocol of~\cref{proto:abspubverif} adds public verifiability to our dynamic
proof of retrievability.
Indeed, LIP security states that in a group
$\G$ of prime order,
the values $(g^{P_1(s)},\ldots,g^{P_m(s)})$ are indistinguishable from a
random tuple of the same size, when $P_1,\ldots,P_m$ are linearly
independent multivariate polynomials of bounded degree and $s$ is
the secret.
Therefore, in our modified protocol, each row
$g^{\UU_i}=\left(g^{\svec_i^j}\right)_{j=1..m}$ is indistinguishable from a random tuple of
size $m$ since  the polynomials $X^j$, $j=1..m$ are independent
distinct monomials.
Then the idea is to reduce breaking the public verifiability to
breaking a discrete logarithm. For this, the discrete logarithm to break
will be put inside $\UU$.

These modifications  give rise to the
following~\cref{thm:LIPpubverif}.
Compared to
\cref{thm:audit}, this requires the LIP security assumptions and a
larger domain of the elements.

\begin{restatable}{theorem}{LIPpubverif}
\label{thm:LIPpubverif}
  Under LIP security in a group $\G$ of
  prime order \(p\ge\max\{16n + 96\statsec,m2^{2\compsec}\}\), where
  discrete logarithms are hard to compute,
  the Protocol of~\cref{proto:abspubverif}
	satisfies correctness, {public} authenticity and retrievability, {as defined in \cref{def:corr,def:pubauth,def:retr}}.
\end{restatable}
\begin{proof}
In~\cref{proto:abspubverif}, Correctness is just to verify the
dotproducts, but in the exponents; this is:
$\PK^y = g^{\UU{y}} = g^{\UU\MM{x}} = \WW^x$.

Public verifiability is guaranteed as $\PK$ and $\UU$, as well as the
roots $r_\MM$ and $r_\ww$ of the Merkle trees for $\MM$ and $\WW$, are
public.
Now for Authenticity: first, any incorrect $\WW$ is detected by the Merkle hash tree
verification.
Second, with a correct~$\WW$, any incorrect~$y$ is also detected with
high probability, as shown next.
Suppose that there exist an algorithm ${\mathcal{A}(\MM,\PK,\WW,r)}$ that
can defeat the verification with a fake $y$, with probability $\epsilon$.
That is the algorithm produces $\bar{y}$, with
$\bar{y}\neq y=\MM{x}$, such that we have the $t$ equations:
\begin{equation}\label{eq:attack} \PK^y = \WW^x = \PK^{\bar{y}}. \end{equation}

We start with the case $t=1$.
Let $A=g^a$ be a DLOG problem.
Then we follow the proof
of~\cite[Lemma~1]{Elkhiyaoui:2016:ETPV} and simulate \Init{} via the following inputs
to the attacker:
\begin{itemize}[nosep]
\item $r\random{S\subseteq\Z^*_p}$ and let $x=[r,r^2,\ldots,r^{{\nright}}]^\intercal$;
\item Sample $\MM\random{S^{\mleft{\times}\nright}}\subseteq{\Z_p^{\mleft\times\nright}}$ and
  $\UU\random{S^{\mleft}\subseteq\Z_p^{\mleft}}$.
\item Randomly select also $k\in 1..\mleft$ and, then,
compute $\PK=g^\UU A^{\vect{e_k}}$, so that
$\PK=g^{\UU+a\vect{e_k}}$, where $\vect{e_k}$ is the $k$-th canonical
vector of $\Z_p^\mleft$.
\item Under LIP security~\cite[Theorem~3.1]{Abdalla:2015:PRF},
  $\PK$ is indistinguishable from the distribution of the protocol ($g^{s_i^j}$).
\item finally compute $\WW=\PK^\MM$, thus also indistinguishable from
  the distribution of the protocol.
\end{itemize}

To simulate any number of occurences of \Write{}, it is then
sufficient to randomly select $\MM'_{ij}$.
Then compute and send to the attacker:
$\WW'_{1..t,j} \leftarrow \WW_{1..t,j} \cdot{}
K_{1..t,i}^{\MM'_{ij}-\MM_{ij}}$
(since $g^\Delta=g^{(\MM'_{ij}-\MM_{ij})\UU_{1..t,i}}=K_{1..t,i}^{\MM'_{ij}-\MM_{ij}}$).

After that, the attacker answers an \Audit{}, with $\bar{y}\neq{y}$
satisfying~\cref{eq:attack}.
This is $g^{(\UU+a\vect{e_k})\bar{y}}=g^{(\UU+a\vect{e_k})\MM{x}}$, equivalent to:
\begin{equation} (\UU+a\vect{e_k})(\bar{y}-\MM{x})\equiv{0}\mod{p}. \end{equation}

Since $\bar{y}\neq{y}\mod{p}$, then there is at least one index $1\le{j}\le{m}$
such that $\bar{y}_j\neq{y_j}\mod{p}$.
Since $k$ is randomly chosen from $1..m$, the probability that
$\bar{y}_k\neq{y_k}\mod{p}$ is at least $1/m$.
If this is the case then with $z=\bar{y}-y$, we have $z_k\neq{0}\mod{p}$
and $\UU{z}+a z_k \equiv{0}\mod{p}$, so that
$a \equiv-z_k^{-1}\UU{z}\mod{p}$.
This means that the discrete logarithm is broken with advantage
$\ge\epsilon/m$.

Finally for any $t\geq{1}$ the proof is similar except that $A$ is put
in different columns for each of the $t$ rows of $\UU$. Thus the probability to
hit it becomes $\ge{t/m}$ and the advantage is
$\ge{t\epsilon/m}\ge\epsilon/m$. This gives the requirement that
$p\ge{m2^{2\compsec}}$ to sustain the best generic algorithms for DLOG.

Retreivability comes from the fact that $y$ and $x$ are public
values. Therefore this part of the proof is identical to that of
\cref{thm:audit}.
\end{proof}

\begin{remarks}
    We mention a few small performance and implementation notes:
\begin{itemize}[nosep]
\item
  If a Writer wants to perform an audit, she does not need to use the
  encrypted control vector $\PK$, nor to store it.
  She just computes~$\UU{y}$ directly, then checks
  that~$g^{\UU{y}}\checks{=}\WW^x$.
\item Even if~$\UU$ is structured,~$\PK$ hides this structure and
  therefore requires a larger storage.
  But any Verifier can just fetch it and~$r_\WW$ from the
  authenticated channel (for instance, electronically signed), as well
  as fetch~$\WW$ from the Server, and perform the verification on the
  fly. Optimal communications for the Verifier are then when
  $m=n=O(\sqrt{N/\log p})$.
\item To save some constant factors in communications, sending~$\WW$ or
  any of its updates~$\WW'_{i,j}$ is not mandatory anymore: the Server
  can now recompute them directly from~$\MM$, $\PK$ and~$\MM'$.
\end{itemize}
\end{remarks}

{In terms of performance,
the most significant changes between the private and public modes
are for the
Verifier's and (to a much lesser extent) server's computation time
during \Audit{}s:
we show in~\cref{sec:impl} that public verification is
more expensive but this remains doable in a few seconds
even on a constrained device.
}

\section{Experiments with Google cloud services}
\label{sec:impl}
As we have seen, compared to other dynamic PoR schemes,
our protocol aims at achieving the high security guarantees of
PoR, while trading near-minimal persistent server storage for increased
audit computation time.

In order to address the practicality of this tradeoff, we implemented
and tested our PoR protocol
using virtual machines and disks on the
Google Cloud Platform service.

Specifically, we address two primary questions:
\begin{itemize}[nosep]
  \item What is the monetary cost and time required to perform our
    \(O(N)\) time audit on a large database?
  \item How does the decreased cost of persistent storage trade-off with
    increase costs for computation during audits?
\end{itemize}

Our experimental results are summarized in
\cref{table:comm,table:results,table:mpi}.
{%
For a 1TB data file,
the \(O(\sqrt{N})\) communication cost of our audit entails less than
6MB of data transfer, and our implementation executes the \(O(N)\)
audit for this 1TB data file in less than $5$ minutes for a monetary
cost of about $8$ cents USD.
}

By contrast, just the extra persistent storage required by other existing PoR schemes
would cost at least \$40 USD or as much as \$200 USD per month, not
including any computation costs for audits.
These
results indicate that
the
communication and computation costs of our \Audit{} algorithm
are not prohibitive in practice despite their unfavorable
asymptotics; and furthermore,
our solution is the most cost-efficient PoR scheme available
when few audits are performed per day.

We also emphasize again that a key benefit to our PoR scheme is its
\emph{composability} with existing software, as the data file is left
intact as a normal file on the Server's filesystem.

The remainder of this section gives the full details of our
implementation and experimental setup.

\subsection{Parameter selection}
\label{sec:twoprime}

Our algorithm treats the database as if it is an~$n \times m$ matrix with
elements in a finite field. {As seen in~\cref{sec:formal}},
we
need the field size to be at least 40 bits or more in order to ensure
authenticity.

{We ran the experiments with two modes: a private one with a
  $57$-bits prime and a public one with a $253$-bits prime}.

In order to maximize the block size while avoiding costly
multiple-precision arithmetic, we used the largest 57-bit prime,
$p=144115188075855859$ {for the private mode}. This allows the input to be read in 7-byte
(56-bit) chunks with no conversion necessary; each 56-bit chunk of raw
data is treated as an element of~$\F_p$. At the same time,
because~$p$ is less than 64-bit, no computations require
multiple-precision, and multiple multiplications can be accumulated in
128-bit registers before needing to reduce modulo~$p$.
Finally, choosing a prime close to (but less than) a power of 2 makes
randomly sampling integers modulo~$p$ especially efficient
{(discarding sampled values larger than~$p$ will seldom
  happen).}

To balance the bandwidth (protocol communications) and the client
computation costs, we represent \MM{} as a square matrix with
dimensions
$m=n=\sqrt{N/56}$, where the 56 comes from our choice of~$\F_p$.
We also fixed the Merkle tree block size at 8KiB for all experiments
and used SHA-512/224 for the Merkle tree hash algorithm.

For the public mode, we used the following libraries%
\footnote{
\url{https://gmplib.org},
\url{https://github.com/linbox-team/givaro},
\url{http://www.openblas.net},
\url{https://linbox-team.github.io/fflas-ffpack},
\url{https://download.libsodium.org}.
}:
\texttt{gmp-6.2.1} and \texttt{givaro-4.1.1} for arbitrary precision
prime fields, \texttt{openblas-0.3.15} and \texttt{fflas-ffpack-2.4.3}
for high-performance linear algebra,
and \texttt{libsodium-1.0.18} for the elliptic curve.
We ran the experiments with
\emph{ristretto255}, a $253$-bits prime order {subgroup of}
\emph{Curve25519}.
We still use a square matrix database, but now with
$m=n=\sqrt{N/252}$.
Depending on the database size, \cref{thm:LIPpubverif} shows that
the computational security parameter of our next experiments is
set to slightly less than $128$ (namely between $117.78$ and $120.2$).

The resulting asymptotic costs for these parameter choices,
{in both modes,} are
summarized in Table~\ref{tab:sqrt}.

\begin{table}[htbp]\centering\small
\caption{Proof of retrievability via square matrix verifiable
  computing}\label{tab:sqrt}
\begin{tabular}{|c|c|c|c|c|}
\hline
\multicolumn{2}{|c|}{}& Server & Comm. & Client\\
\hline
\multicolumn{2}{|c|}{Storage} & $N+o(N)$ & & $O(\sqrt{N})$ \\
\hline
\multirow{3}{*}{\rotatebox[origin=c]{90}{Comput.}} &\Init{} & $O(N)$ & $N$ & $O(N)$ \\
&\Audit{} & $O(N)$ & $O(\sqrt{N})$ & $O(\sqrt{N})$  \\
&\Read{}/\Write{} & $O(\log(N))$ & $O(\log(N))$& $O(\log(N))$  \\
\hline
\end{tabular}
\end{table}

\subsection{Experimental Design}
Our implementation provides the \Init{}, \Read{}, \Write{}, and \Audit{}
algorithms as described in the previous sections, including the Merkle hash
tree implementation for read/write integrity. As the cost of the first three
of these are comparable to prior work, we focused our experiments
on the \Audit{} algorithm.

We ran two sets of experiments, using virtual machines and disks on
Google Cloud's Compute
Engine\footnote{\url{https://cloud.google.com/compute/docs/machine-types}.}.

\begin{table}[htbp]\centering\small
\caption{Google Cloud Server VMs}\label{table:vms}
{\footnotesize
Costs as of May 2021. Each physical core is hyperthreaded as two vCPUs.
}
\setlength\tabcolsep{4pt} %
\renewcommand{\arraystretch}{1.1} %
\begin{tabular}{|c|c|c|c|c|c|} \hline
\multirow{2}{*}{} & \multirow{2}{*}{Family} & Physical & Main & Local & Cost/hour \\
&& Cores & Memory & SSDs & (USD) \\ \hline
\clientinstance{} & {\footnotesize f1-micro} & 0.1 & 0.6 GB & --- & \$0.01 \\
\serva{} & {\footnotesize n1-standard-2} & 1 & 7.5 GB & 1.5TB & \$0.26 \\
\servb{} & {\footnotesize n1-standard-8} & 4 & 30 GB & 6TB & \$1.04 \\
\servc{} & {\footnotesize n1-standard-32} & 16 & 120 GB & 9TB & \$2.51 \\
\hline
\end{tabular}
\end{table}

The client machine \clientinstance{} was a cheap f1-micro instance
with a shared vCPU
and low RAM, in the europe-west1 region.
For the server, we used 3 different VMs $\serva, \servb, \servc$
as listed in \cref{table:vms}, all in the us-central1 region.
The database file itself was stored on local SSD drives for maximal
throughput in our audit experiments. {Although our implementation does
not use more disk space than the size of the database plus the size of
the Merkle tree (never more than 1.007TB in our experiments), we
over-provisioned the SSDs to achieve higher throughputs; the prices in
\cref{table:vms} reflect the total VM instance and storage costs.}

For testing, we generated
files of size 1GB, 10GB, 100GB, and 1TB
filled with random bytes.
All times reported are total ``wall-clock'' time unless otherwise noted.

Except where noted with an asterisk (*), {where
  experiments were run only once,} all values are the median over
11 runs, ignoring the first run in order to ``warm up'' caches etc.
Note that this actually had a significant effect on the larger machine size
which also has more RAM, as the 16-core machine can cache all sizes except the
1TB database in memory.

\subsection{Audit compared to checksums}

For the first set of experiments, we
wanted to address the question of how ``heavy'' the hidden
constant in the~$O(N)$~is.
For this, we compared the cost of performing
a single audit, on databases of various sizes, to the cost of computing
a cryptographic checksum of the entire database using the standard
Linux checksum tools \texttt{md5sum} and \texttt{sha256sum}.

\begin{table}[htbp]\centering \small
\caption{Single-threaded experiments on Google
  Cloud}\label{table:results}
{\footnotesize
Values indicate the median number of seconds for a single run
on the \serva{} machine.
Except where noted with (*), each experiment was performed 11 times.
In all cases, after discarding at most one outlier value, the maximum relative
difference between the runs was at most 20\%.}\\
\renewcommand{\arraystretch}{1.1} %
\begin{center}
\begin{tabular}{|c||c|c|c|c|} \hline
\multicolumn{1}{|c||}{Operation} &      1GB     &       10GB    &
100GB   &       1TB\\   \hline
\multirow{1}{*}{MD5}
& 1.87 & 20.58 & 202.51 & 2017.76* \\
\multirow{1}{*}{SHA256}
& 5.21 & 54.52 & 561.22 & 5413.35* \\
\multirow{1}{*}{\Init{}}
& 2.46 & 29.42 & 284.75* & 2772.14* \\
\hline \multicolumn{5}{c}{PRIVATE-VERIFIED AUDIT USING 57-BIT PRIME} \\
\hline
\multirow{1}{*}{Client}
& 0.00 & 0.00 & 0.00 & 0.01 \\
\multirow{1}{*}{Server}
& 0.24 & 4.93 & 53.05 & 529.90 \\
\hline \multicolumn{5}{c}{PUBLIC-VERIFIED AUDIT USING RISTRETTO255} \\
\hline
\multirow{1}{*}{Client}
& 0.53 & 1.67 & 5.37 & 16.81 \\
\multirow{1}{*}{Server}
& 1.65 & 17.1 & 173.49 & 1725.75* \\
\hline
\end{tabular}
\end{center}
\end{table}

In a sense, a cryptographic checksum is another means of integrity check
that requires no extra storage, albeit without the malicious server
protection that our PoR protocol provides. Therefore, having an audit
cost which is comparable to that of a cryptographic checksum indicates
the~$O(N)$ theoretical cost is not too heavy in practice.

\Cref{table:results} confirms that the cost of our \Audit{} procedure
scales linearly with the database size, as expected.
Furthermore, we can see that audits are very efficient in practice,
being even faster than the built-in checksum utilities in our tests.
That is, the $O(N)$ running time of our \Audit{} algorithm is actually
feasible, for both the private and public mode, even at the terabyte
scale.
The public mode is slightly slower, as expected in~\cref{rq:fieldsize}.

\subsection{Parallel server speedup for audits}
Our experimental results in \cref{table:results} indicate good
performance for our \Audit{} algorithm, but at the larger end of
database sizes such as 1TB, the $O(N)$ work performed by the server
still incurs a significant delay of several minutes.
To demonstrate that a more powerful server can handle large-size
\Audit{}s even more efficiently,
we used OpenMP to parallelize the main loop of
our \Audit{} algorithms.
These routines are trivially
parallelizable{:} each parallel
  core
performs the matrix-vector product on a contiguous
subset of rows of~\MM{}, corresponding to a contiguous segment of the underlying file.

Because the built-in MD5 and SHA256 checksum programs do not achieve any
parallel speedup, we focused only on our \Audit{} algorithm for this set
of experiments.
The results are reported in \cref{table:mpi}.

When the computation is CPU-bound, as is the case mostly with the public
verified version that uses larger primes, CPU utilization is high
and we achieve linear
speedup compared to the single-core timings in \cref{table:results}.
For the more efficient 57-bit private verification version, the
speedup compared to \cref{table:results} is sometimes more and sometimes
less than linear, for two reasons that have to do with the I/O bottleneck
between disk and CPU.

First, the larger machines \servb{} and \servc{} that are used here do
not just have more cores than \serva{}; they also have more RAM and more
(over-provisioned) local SSD space. This allows \servb{} to
entirely cache the 10GB database and \servc{} to entirely cache
the 10GB and 100GB databases, leading to sometimes super-linear speedup
when the computation is I/O-bound.

The second observation is that, even using the fastest solution
available (local SSDs) in Google Cloud, we could not achieve
greater than roughly 4GB/sec throughput reading from disk.
This effectively creates a ``maximum speed'' for any computation,
which limits the benefit of additional cores especially for the 1TB
audit with the small 57-bit prime.

To a lesser extent these two phenomena also occur in the public mode.
There, they are however partially compensated by a better parallelism
pertaining an increase in the computations.

However, we emphasize again that this is a \emph{good thing} --- our
\Audit{} algorithm is efficiently parallelizable, up to the inherent
limiting speed of fetching data from the underlying storage.

We also used these times to measure the total cost of running each audit in Google Cloud Platform,
which features per-second billing of VMs and persistent disks, as reported in \cref{table:mpi}
as well.
{Interestingly, due to the disk throughput limitations
discussed above, the 4-core VM is more cost-effective for
private-verified audits.}

\begin{table}[htbp]\centering\small
\caption{Multi-core server times for Audit}\label{table:mpi}
{\footnotesize
Values indicate the median number of seconds wall-time for a single run.
Except where noted with (*), each experiment was performed 11 times.
In all cases, after discarding at most one outlier, the maximum relative
difference between the runs was at most 20\%.}\\
\setlength\tabcolsep{5pt} %
\renewcommand{\arraystretch}{1.1} %
\begin{tabular}{|c|c||c|c|c|c|} \hline
Server & Metric &      1GB     &       10GB    &       100GB   &
1TB\\   \hline
\multicolumn{6}{c}{PRIVATE-VERIFIED AUDIT USING 57-BIT PRIME} \\
\hline
\multirow{2}{*}{\servb{}}
& Audit & 0.06 & 0.62 & 29.08 & 278.37 \\
\cline{2-6}
& Cost & \$0.00002 & \$0.0002 & \$0.008 & \$0.08 \\
\hline
\hline
\multirow{2}{*}{\servc{}}
& Audit & 0.03 & 0.22 & 1.88 & 250.91 \\
\cline{2-6}
& Cost & \$0.00002 & \$0.0002 & \$0.001 & \$0.175 \\
\hline
\multicolumn{6}{c}{PUBLIC-VERIFIED AUDIT USING RISTRETTO255} \\
\hline
\multirow{2}{*}{\servb{}}
& Audit & 0.45 & 4.37 & 51.45 & 536.09* \\
\cline{2-6}
& Cost & \$0.0001 & \$0.001 & \$0.015 & \$0.155 \\
\hline
\hline
\multirow{2}{*}{\servc{}}
& Audit & 0.12 & 1.21 & 11.87 & 357.49* \\
\cline{2-6}
& Cost & \$0.0001 & \$0.001 & \$0.008 & \$0.249 \\
\hline
\end{tabular}
\end{table}

\subsection{Network communication costs}
Having closely examined the server and client computation times,
we finally turn to the \(O(\sqrt{N})\) communication bandwidth between
client and server during audits.
Recall that our client \clientinstance{} was located in western Europe and the servers
\serva, \servb, \servc{} were located in central North America.
As a baseline, we used ping and scp to determine the
client-server network connection: it had an average round-trip latency of
101ms and achieved throughput as high as 19.1 MB/sec.

The time spent communicating the challenge and response vectors, \xx{} and
\yy{}, becomes insignificant in comparison to the server computation as the size of the database increases.
In the case of our experiments, \cref{table:comm} summarizes that
communication time of both \xx{} and \yy{} remains under
{two} seconds.
  We also list the total amount of data communicated, which exhibits
square root scaling as expected.

\begin{table}[htbp]\centering\small
\caption{Amount of Communication Per \Audit}\label{table:comm}
{\footnotesize
Values indicate the median number of seconds for a single run with the
\servb{} server.
Each experiment was performed 11 times, with a maximum variance of 13\%
between runs.}\\
	\begin{tabular}{|c||c|c|c|c|} \hline
		Metric & 1GB & 10GB & 100GB & 1TB \\ \hline \hline
		Comm. (KB) & 187 & 591 & 1868 & 5906 \\ \hline
		Time (s) & 0.73 & 1.19 & 1.50 & 1.80 \\ \hline
	\end{tabular}
\end{table}

\section{Detailed state of the art}
\label{sec:sota}
PDP schemes,  first introduced by {Ateniese et al.} \cite{ateniese2007provable},
originally only considered static data storage.
The original scheme was later adapted to
allow dynamic updates by {Erway et al.} \cite{Erway} and has since seen numerous
performance improvements. However, PDPs only guarantee
(probabilistically) that a
\emph{large fraction} of the data was not altered; a single block
deletion or alteration is likely to go undetected in an audit.

PoR schemes, independently introduced by
{Juels et al.} \cite{juels2007pors}, provide a stronger guarantee of integrity: namely, that
any small alteration to the data is likely to be detected. In this
paper, we use the term PoR to refer to any scheme which provides this
stronger level of recoverability guarantee.

PoR and PDP are usually constructed as a collection of phases in order
to initialize the data storage, to access it afterwards and to audit
the server's storage. Dynamic schemes also propose a modification of
subsets of data, called write or update.
Since 2007, different schemes have been proposed  to serve different
purposes such as data confidentiality,
data integrity, or data availability, but also  freshness and fairness.
Storage efficiency, communication efficiency and reduction of disk I/O
have improved with time.
Some schemes  are developed for static data (no update algorithm) ,
others extend their audit algorithm  for public verification, still
others require a finite number of Audits and Updates.

\subsection{Low storage overhead}\label{ssec:storj}
 The schemes of Ateniese et al. \cite{ateniese2007provable} or Seb\'e
et al. \cite{Sebe:2008:EfficientRD} are in the PDP model.
Both of them have a storage overhead in $o(N)$. They use the
RSA protocol in order to construct homomorphic authenticators, so that a
successful
audit guaranties data possession on some selected blocks.  When
all the blocks are selected, the audit is deterministic but the
computation cost is high. So in practice, \cite{ateniese2007provable}
minimizes the file block accesses, the computation on the server, and
the client-server communication.  For one audit on at most $f$
blocks,the S-PDP protocol of \cite{ateniese2007provable} gives the
costs seen in \cref{tab:PSPDP}.  A robust auditing integrates S-PDP with a
forward error-correcting codes to mitigate arbitrary small file
corruption. Nevertheless, if the server passes one audit, it guarantees
only that a portion of the data is correct.

\begin{table}[htbp]\centering\small
\caption{ S-PDP on $f$ blocks :
The file $M$ is composed of $N/\mtblock$ blocks of bit-size
$\mtblock$.}\label{tab:PSPDP}
{\footnotesize The computation is
made mod Q, a product of two large prime numbers.}
\renewcommand{\arraystretch}{1.1}
\begin{tabular}{|c|c|c|cc|}
\hline
\multicolumn{2}{|c|}{}& Server & Communication & Client\\
\hline
\multicolumn{2}{|c|}{Storage} & $N+ m $ & & $ O(1)$ \\
\hline
\multirow{3}{*}{\rotatebox[origin=c]{90}{Comput.}} &Setup &  & $  N+f $ & $O(\mtblock f)$ \\
&Audit & $O(f)$ & $ O(1)$ & $O(f)$ \\
 &&&&\\
\hline
\end{tabular}
\end{table}

Later,  Ateniese et al. \cite{ateniese2008scalable} proposed a
scheme secure under the
random oracle model based on hash functions and symmetric keys. It
has an efficient update algorithm but uses tokens which impose a
limited number of audits or updates.

Alternatively, verifiable computing can be used to go through the whole
database with Merkle hash trees, as in~\cite[\S
6]{Benabbas:2011:vdcld}. The latter proposition however comes with a
large overhead in homomorphic computations and does not provide an
Audit mechanism. Verifiable computing can provide an audit mechanism,
as sketched by {Fiore and Gennaro} in~\cite{Fiore:2012:PVD}, but then it
is not dynamic anymore.

Storj
\cite{Storj:2016:por} (version 2) is a very different approach also based
on Merkle hash trees. It is a dynamic PoR protocol with bounded Audits.
The storage is encrypted and cut into $m$ blocks of size $\mtblock$.
For each block and for a selection of $\sigma$ salts, a Merkle Hash tree
with $\sigma$ leaves
is constructed.
The efficiency  of Storj is presented Table~\ref{tab:Storj}.
Storj allows only a fixed number of audits (the number of seeds)
before the entire data must be re-downloaded to restart the
computation. This is a cost of $O(N\sigma)$ operations for the client
every $\sigma$ audits, and thus an average cost of $O(N)$. Our PoR
supports unlimited and fast audits, of cost always $O(\log{n})$.

\begin{table}[htbp]\centering\small
\caption{ Storj-V2:
The file $M$ is composed of $N/\mtblock$ blocks of bit-size~$\mtblock$. $\sigma$ is the number of salts.}\label{tab:Storj}
\setlength{\tabcolsep}{3pt}
\renewcommand{\arraystretch}{1.1}
\begin{tabular}{|c|c|c|c|c|}
\hline
\multicolumn{2}{|c|}{}& Server & Comm. & Client\\
\hline
\multicolumn{2}{|c|}{Storage} & $N{+}O( \frac{N}{\mtblock}\sigma)$ & & $O(\frac{N}{\mtblock} \sigma)$\\
\hline
\multirow{3}{*}{\rotatebox[origin=c]{90}{Comput.}}
& Setup &  & $ N{+}O(\frac{N}{\mtblock} \sigma)$ & $O(N\sigma)$ \\
& Avg. Audit & $O(N+\frac{N}{\mtblock} \sigma)$ &   $O(\frac{N}{\mtblock}\log \sigma+ \frac{N}{\sigma})$  & $O(N)$ \\
& Update &  & $ \mtblock {+} O(\sigma)$ & $O(\mtblock\sigma)$  \\
\hline
\end{tabular}
\end{table}

\subsection{Fast audits but large extra storage}
PoR methods based on block erasure encoding are a class of methods
which guarantee with a high probability that  the client's entire data
can be retrieved.  %
The idea is to check the authenticity of a number of erasure encoding
blocks during the data recovery step but also during the audit
algorithm. Those approaches will not detect a small amount of corrupted
data. But the idea is that if there are very few corrupted blocks,
they could be easily recovered via the error correcting code.

Lavauzelle et al., \cite{Lavauzelle:2016:ldcpor} proposed a static
PoR.
The \Init{} algorithm consists in  %
encoding the file %
using a lifted  q-ary Reed-Solomon code %
and  encrypting it with a block-cipher.
The Audit  algorithm checks if one word of $q$ blocks belongs to a set
of Reed-Solomon code words. This algorithm has to succeed a sufficient
number of times to ensure with a high probability that the file can be
recovered. Its main drawback is that it requires an initialization
quadratic in the database size. For a large data file of several
terabytes this becomes intractable.

In addition to a block erasure code, PoRSYS of Juels et
al. \cite{juels2007pors} use block encryptions and  sentinels in
order to store static data with a cloud server. %
Shacham and Waters \cite{shacham2008compact} use authenticators to
improve the audit algorithm. A publicly verifiable scheme based on the
Diffie-Hellman problem in bilinear groups is also proposed.

Stefanov et al. \cite{stefanov2012iris} were the first to consider a
dynamic PoR scheme.
Later improvements by Cash et al. or Shi et
al. \cite{Cash:2017:DPR,Shi:2013:orampor} allow for dynamic updates
and reduce the asymptotic complexity (see Table~\ref{tab:Shi}). However, these techniques rely on
computationally-intensive tools, such as locally decodable codes and
Oblivious RAM (ORAM), and incur at least a 1.5x, or as much as 10x,
overhead on the size of remote storage.

\begin{table}[htbp]\centering\small
\caption{ Shi et al.~\cite{Shi:2013:orampor}:
The file $M$ is composed of $\frac{N}{\mtblock}$ blocks of bit-size~$\mtblock$.}\label{tab:Shi}
\renewcommand{\arraystretch}{1.1}
\begin{tabular}{|c|c|c|c|c|}
\hline
\multicolumn{2}{|c|}{}& Server & Communication & Client\\
\hline
\multicolumn{2}{|c|}{Storage} & $O( N)$ & & $ O(\mtblock)$\\
\hline
\multirow{3}{*}{\rotatebox[origin=c]{90}{Comput.}} &Setup &  &$
 N+O(\frac{N}{\mtblock})$  & $O(N \log N)$  \\
& Audit & $O(\mtblock  \log N)$ & $%
  O(\mtblock + \log N)$& $ O(\mtblock + \log N)$ \\
&Update & $O(\mtblock \log N)$ & $%
 O( \mtblock+ \log N)$ & $O(b+ \log N)$  \\

\hline
\end{tabular}
\end{table}

Recent variants include {\em Proof of Data Replication} or {\em Proof
  of Data Reliability}, where the error correction is performed by the
server instead of the
client~\cite{Armknecht:2016:mirror,Vasilopoulos:2019:portos}.
Some use a weaker, {\em rational}, attacker
model~\cite{Moran:2019:PoST,Cecchetti:2019:PIE},
and in all of them the client thus has to also be able to
verify the redundancy; but we do not know of dynamic versions of these.

\begin{table}[htbp]\centering\small
\caption{Comparison of our low server storage protocol with that of Shi et al.~\cite{Shi:2013:orampor}.}\label{tab:protocol}
\begin{tabular}{p{2.475cm}p{1.6cm}p{1.35cm}p{1.35cm}}
\toprule
&Shi  & Here& Here \\
&et al.~\cite{Shi:2013:orampor}& \ExternT & \ExternF \\
\midrule
Server extra-storage & $5N$                 & $o(N)$         & $o(N)$ \\
Server audit cost    & $O(\mtblock\log N)$ & $N{+}o(N)$       & $N{+}o(N)$ \\
Communication        & $O(\mtblock{+}\log N)$ & $O(\sqrt{N})$ & $O(N^\alpha)$    \\
Client audit cost    & $O(\mtblock{+}\log N)$ & $O( \sqrt{N})$ & $O(N^{1{-}\alpha})$ \\
Client storage       &  $O(\mtblock)$  & $O(1)$ & $O(N^{1{-}\alpha})$ \\
\bottomrule
\end{tabular}
\end{table}
Table~\ref{tab:protocol} compares the additional server storage and audit costs
between~\cite{Shi:2013:orampor} and the two variants of our protocol:
the first one saving on
communication, and the
second one, externalizing the storage of the secret audit matrix $V$.
In the former case, an arbitrary parameter $\alpha$ can be used
in the choice of the dimensions: $m=N^\alpha$ and $n=N^{1-\alpha}/\log_2(q)$.
This balances between the
communication cost $O(N^\alpha$)
and the Client computation and storage $O(N^{1-\alpha})$.

Note that efficient solutions to PoR  for dynamic data do not
consider the confidentiality of the file $M$, but assume that the user
can encrypt its data in a prior step if needed.

\section{Conclusion}
We presented new protocols for dynamic Proof of Retrievability, based
on randomized linear algebra verification schemes over a finite
field.
Our protocols do not require any encoding of the database
and are therefore near optimal in terms of persistent storage on the server
side. They include also efficient unlimited partial {retrievals} and
updates as well as provable {retrievability} from malicious servers.
They are implementable with simple cryptographic building blocks and are
very efficient in practice as shown for instance on a Google Compute
platform instance.
With the addition of any IND-CPA symmetric cipher the clients become
nearly stateless; adding a group where the discrete logarithm is hard
also enables a public verification.

On the one hand, private proofs are very fast, less than a second on
constrained devices. On the other hand, while still quite cheap, the
public verification could nonetheless be improved. Precomputations of
multiples of elements of $\PK$ and $\UU$, combined with dedicated
methods for dotproduct in the exponents
(generalizing of Shamir's trick for simultaneous exponentiations)
might improve the running time.
More generally, our verification is a dotproduct, or a polynomial
evaluation when the control vectors are structured. This verification
itself could be instead computed on the {server} side and only verified
by a client, using for instance succinct
non-interactive arguments of knowledge.
\section*{Availability}

The source code and script to perform the experiments of~\cref{sec:impl} are
available via the following {GitHub} repository:
\url{https://github.com/dsroche/la-por}.

\providecommand{\noopsort}[1]{}

\end{document}